\documentclass[letter,12pt]{article}
\usepackage{amsfonts, amsmath, amsthm, amssymb}
\usepackage{color}
\usepackage{bbm}
\usepackage{graphicx}
\usepackage{graphics}
\usepackage{multirow}
\usepackage{epstopdf}
\usepackage{float}
\usepackage{subfig}
\usepackage[countmax]{subfloat}
\usepackage{rotating}
\usepackage{fullpage} % NOTE: Can not use this with revtex 4 (makes 2000 page documents - not cool.)
\usepackage{setspace}
 \usepackage{xcolor}
\usepackage{authblk}
\usepackage{xr}

\newcommand{\rEcomment}[1]{}
\newcommand{\rYcomment}[1]{}
\newcommand{\rJcomment}[1]{}
\newcommand{\rKcomment}[1]{}

\newcommand{\bfA}{\mathbf{A}}
\newcommand{\bfB}{\mathbf{B}}
\newcommand{\bfC}{\mathbf{C}}
\newcommand{\bfD}{\mathbf{D}}
\newcommand{\bfE}{\mathbf{E}}

\newcommand{\bfH}{\mathbf{H}}
\newcommand{\bfI}{\mathbf{I}}

\newcommand{\bfM}{\mathbf{M}}

\newcommand{\bfP}{\mathbf{P}}

\newcommand{\bfS}{\mathbf{S}}

\newcommand{\bfU}{\mathbf{U}}

\newcommand{\bfW}{\mathbf{W}}
\newcommand{\bfX}{\mathbf{X}}

\newcommand{\bfa}{\mathbf{a}}
\newcommand{\bfb}{\mathbf{b}}

\newcommand{\bfd}{\mathbf{d}}

\newcommand{\bfu}{\mathbf{u}}

\newcommand{\bfx}{\mathbf{x}}
\newcommand{\bfy}{\mathbf{y}}

\newcommand{\At}{\tilde{A}}

\newcommand{\yt}{\tilde{y}}

\newcommand{\bfPt}{\mathbf{\tilde{P}}}

\newcommand{\bfxt}{\mathbf{\tilde{x}}}
\newcommand{\bfyt}{\mathbf{\tilde{y}}}

\newcommand{\om}{\omega}

\newcommand{\bfLambda}{{\boldsymbol\Lambda}}
\newcommand{\bfTheta}{{\boldsymbol\Theta}}

\newcommand{\bfxi}{{\boldsymbol\xi}}

\newcommand{\beq}{\begin{equation}}
\newcommand{\eeq}{\end{equation}}
\newcommand{\beqr}{\begin{eqnarray}}
\newcommand{\eeqr}{\end{eqnarray}}
\newcommand{\beqrn}{\begin{eqnarray*}}
\newcommand{\eeqrn}{\end{eqnarray*}}
\newcommand{\beqn}{\begin{equation*}}
\newcommand{\eeqn}{\end{equation*}}
\newcommand{\bei}{\begin{itemize}}
\newcommand{\beii}{\begin{itemize} \item}
\newcommand{\eei}{\end{itemize}}
\newcommand{\ben}{\begin{enumerate}}
\newcommand{\een}{\end{enumerate}}
\newcommand{\bes}{\begin{small}}
\newcommand{\ees}{\end{small}}
\newcommand{\bec}{\begin{center}}
\newcommand{\eec}{\end{center}}
\newtheorem{theorem}{Theorem}[section]

\newtheorem{corollary}[theorem]{Corollary}

\theoremstyle{definition}

\theoremstyle{remark}

\newcommand{\EV}{\mathbf{E}} % Bold E for expected value
\newcommand{\EVb}[1]{\EV\left\{#1\right\}} % EV with fitted  CURLY brackets
\newcommand{\EVs}[1]{\EV\left[#1\right]} % EV with fitted  SQUARE brackets

\newcommand{\K}{{\boldsymbol\kappa}}
\newcommand{\M}{{\boldsymbol\mu}}
\newcommand{\ER}{Erd\"{o}s-R\'{e}nyi}
\newcommand{\WK}{Wiener-Khinchin}
\newcommand{\OU}{Ornstein-Uhlenbeck}
\newcommand{\BA}{Bar\'{a}basi-Albert}

\makeatletter
\renewcommand*{\@fnsymbol}[1]{\ensuremath{\ifcase#1\or \dagger \or \dagger\or \ddagger\or
   \mathsection\or \mathparagraph\or \|\or **\or \dagger\dagger
   \or \ddagger\ddagger \else\@ctrerr\fi}}
\makeatother

\begin{document}

\title{Local paths to global coherence: Cutting networks down to size }
\date{\today}
\author[2]{Yu Hu$^{1}$, James Trousdale}
\author[2,3]{Kre\v{s}imir Josi\'{c}\thanks{These authors contributed equally.}}
\author[1,4,5]{Eric Shea-Brown$^\dagger$}
\affil[1]{Department of Applied Mathematics, University of Washington, Seattle, WA 98195}
\affil[2]{Department of Mathematics, University of Houston, Houston, TX 77204-5001}
\affil[3]{Department of Biology and Biochemistry, University of Houston, Houston, TX 77204-5001}
\affil[4]{Program in Neurobiology and Behavior, University of Washington, Seattle, WA 98195}
\affil[5]{Department of Physiology and Biophysics, University of Washington, Seattle, WA 98195}

\maketitle

% section numbering format
\renewcommand{\thesection}{\Roman{section}} 
\renewcommand{\thesubsection}{\Alph{subsection}}

\begin{abstract}
How does connectivity impact network dynamics? We address this question by linking network characteristics on two scales.  On the global scale we consider the coherence of overall network dynamics. We show that such \emph{global coherence} in activity can often be predicted from the \emph{local structure} of the network.  To characterize local network structure we use ``motif cumulants," a measure of the deviation of pathway counts from those expected in a minimal probabilistic network model.

We extend previous results in three ways.  First, we give a new combinatorial formulation of motif cumulants that relates to the allied concept in probability theory.  Second, we show that the link between global network dynamics and local network architecture is strongly affected by heterogeneity in network connectivity. However, we introduce a network-partitioning method that recovers a tight relationship between architecture and dynamics. Third, for a particular set of models we generalize the underlying theory to treat dynamical coherence at arbitrary orders (i.e. triplet correlations, and beyond). We show that at any order only a highly restricted set of motifs impact dynamical correlations.

\end{abstract}

\clearpage

\section{Introduction}

From genetics to neuroscience to the social 
world, networks of stochastic dynamical systems are ubiquitous.  The architecture of these networks is complex:  irregular but far from random, with an unexpected prevalence of specific connection features~\cite{Bonifazi:09,Song:2005,Perin:2011,Milo:2004jc,Larimer:2008}.  At the same time, networks produce complex patterns of collective dynamics~\cite{Pecora:1998uf,SS00,rinzel89}.  Here we explore the links between these two phenomena and provide  general principles that relate network architecture to collective dynamics.  

%How can the collective dynamics of a network be quantified?   
%
%\emph{The joint activity of pairs and groups of nodes is frequently described using pairwise and higher-order correlations (coherence)~\cite{Renart:2010,Pernice:2011,Pernice:2012,Lindner:2005, Schneidman:2003ej}. Specifically, consider the time dependent activity of a node $x_i(t)$.  Then, the pairwise correlation is $\cov(x_i(t),x_j(t))=\EVb{x_i(t)x_j(t)}-\EVb{x_i(t)}\EVb{x_j(t)}$.  One can also compare activities separated by a certain delay, $\cov(x_i(t),x_j(t+\tau))$, and thus acquire the correlation function (of $\tau$). Another useful generalization is to higher order correlations that consider three or more nodes at a time.  }
%
% \Kcomment{I think the above paragraph is too detailed. First of all, we define covariances after talking about correlations, which 
% is confusing.  We would have to normalize to clear this up, which brings up more problems. I suggest we use the following paragraph.}
% \Ycomment{I agree that the original paragraph has more details than probably needed. I'm fine with the new paragraph, just added a reference to a footnote [35] we already have.}
 
 The joint activity of pairs and groups of nodes is frequently described using pairwise~\cite{corr_def} and higher-order correlations (coherence)~\cite{Renart:2010,Pernice:2011,Pernice:2012,Lindner:2005, Schneidman:2003ej}.
But what do such measures of coherence tell us?  A high average correlation (across all node pairs) reflects approximate synchrony.  In some settings, this global synchrony is what matters for how strongly a network will ``cooperate" to influence a system downstream~\cite{Fries:2005,Sin+95,Salinas:2000,Die+99,piko01}.  Beyond the impact on downstream targets, synchrony can also have an impact on how information is encoded in network activity.  This has been widely studied in the neural networks of sensory pathways, which encode signals from the external world; here, synchronous fluctuations can either serve as a separate ``channel", or can modulate the amount of information that network responses can carry by shaping their overall signal-to-noise ratiocinates~\cite{Aver:2006,gawne93,Cohen:2011eh,Zohary:1994ei,Sompolinsky:2001hh,Abbott:1999ul}. 

\begin{figure}[b!]
\centering
\includegraphics[scale=1.4]{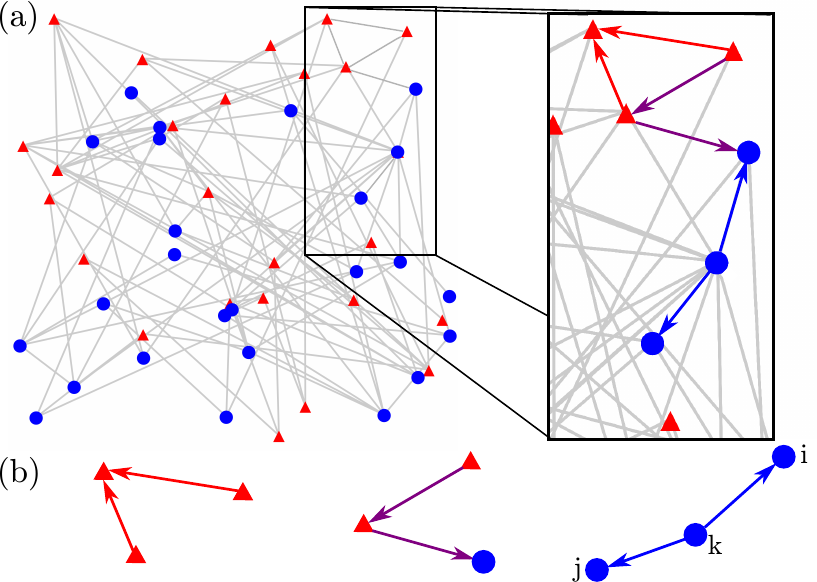}
\renewcommand{\figurename}{FIG.}
\caption{(a) We consider directed, recurrent, networks within which motifs of different orders, i.e. number of edges, are identified. The inset illustrates how these motifs are embedded in the graph. The different shape and color of nodes become relevant and are explained in Sec.~\ref{S:heter_sub}.\ref{S:heter_network}. (b) Left to right:  second order converging, chain, and diverging motifs, as highlighted with color in the inset of (a). Only the latter two contribute to correlations in the path expansion given by Eq.~\eqref{e:cov_epn}. }
\label{F:0}
\end{figure}

We thus turn to the question of relating coherent network dynamics to connectivity structure as described by a directed graph %with weighted edges 
specifying node interactions.  Despite significant progress~\cite{Renart:2010,Ginzburg:1994wm,Sej76,Pernice:2011,Pernice:2012}, this problem remains a challenge. One approach is to identify the key {\it local} connectivity features of a complex network that predict {\it global} levels of correlation ---  the averaged correlation across all nodes in the network. The local connectivity is characterized using specific pathways between subsets of nodes, or motifs. Formally, motifs are particular connectivity patterns (usually smaller graphs) that occur, possibly multiple times, in the graph of the network. Several example network {motifs} are shown in Fig.~\ref{F:0}.

How can motif structure be used to predict network-wide correlation?  An approximate expression relating correlations to the frequency of different types of network motifs has been derived previously~\cite{Pernice:2011,Trousdale:2012}.  Although this result lead to a number of insights, it is difficult to apply generally due to the combinatorial explosion of motifs that appear in the approximation~\cite{Pernice:2011,Trousdale:2012,Hu:2012vh}.  It is necessary to measure empirically the frequencies of many different motifs in order to apply the theory. In earlier work we sought to simplify the situation~\cite{Hu:2012vh}.  We used the frequency of a few, smaller motifs  to predict the frequency of larger motifs in the network.  As as a result we showed that the frequency of a few small motifs alone could predict network-wide correlation --- in many cases with a high level of accuracy.

However, three key questions remain unanswered. First, under what conditions can a set of small motifs be used to accurately infer
the frequency of large motifs?  Second, what features of network connectivity, or motifs, predict higher order correlations?  Third, when our earlier methods fail~\cite{Hu:2012vh} --- that is, when the frequency of small motifs alone does not provide accurate information about correlations --- 
is there a way to  still cut the dynamical complexity down to size?

In this paper we answer these questions. We first summarize, and where necessary reinterpret, our earlier results~\cite{Hu:2012vh} employing new combinatorial definitions: Borrowing ideas from probability theory we define motif moments and cumulants. This abstract approach both reveals the probabilistic structure of our underlying assumptions and allows us to immediately generalize our theory to link higher order correlations in network dynamics to graphical features described by frequencies of more complex motifs. Intriguingly, only a highly restricted set of motifs enter in expressions for dynamical correlations of any given order. We explicitly identify these motifs associated with every order. 
Finally, we apply our method to new types of networks, and show that heterogeneity in network connectivity can lead to a failure of the predictive approach in~\cite{Hu:2012vh}.  
However, even in this case an accurate approximation can be obtained if the network is correctly partitioned, and motif frequencies are
measured within and across the partitions.

Our results for coherence at both second and higher orders hold for stochastic networks where node interactions can be described using linear response, including linear SDEs (Ornstein-Uhlenbeck) and shot noise processes~\cite{Gardiner:1986}
on networks.  Moreover, our findings for second (but not higher) order coherence also hold for coupled point process systems; including networks of integrate--and--fire neurons~\cite{Trousdale:2012}, as well as linearly interacting point processes
(Hawkes models~\cite{Hawkes:1971-2,Pernice:2011}).

\section{Stochastic dynamics on networks}

\subsection{Model of stochastic dynamics on networks}

Stochastic networks of linearly interacting units can generally be described using 
\beq
\label{e:lin}
y_i(t)=x_i(t)+ A_i (t)* \sum_j\bfW_{ij} y_j(t).
\eeq
Here the activity of the $i^{th}$ node, $y_i(t)$, is perturbed linearly from a (stochastic) baseline $x_i(t)$ by filtered input ($*$ stands for convolution) from the rest of the network. The response of unit $i$ is captured by its linear response function $A_i(t)$, and $\bfW_{ij}$ is the connection strength of the input from unit $j$ to unit $i$. An examples of such a stochastic system includes the multivariate \OU{} (OU) process,  widely used to model biological networks~\cite{Pedraza:05,Tomioka:2004bs,Lestas:hy,Warren:2005vh}. We illustrate many of our ideas using this OU process. Details about how the OU process can be put into the form of Eq.~\eqref{e:lin} are in Appendix~\ref{S:OU_relation}, and  details about our numerical results in Appendix~\ref{S:num_details}.  We include a notation summary in Table~\ref{T:notation} in the Appendix.
 
For simplicity, we assume that  connection weights are uniform, so that $\bfW=w\bfW^0$ for an adjacency matrix $\bfW^0$. 
We also assume that the nodes are homogeneous in their dynamics and response to inputs, 
so that $A_i(t)=A(t)$, and $x_i(t)$ are i.i.d. processes.
These assumptions can be relaxed as explained in~\cite{Hu:2012vh}. 

\subsection{Cross-correlation and network motifs}

Our goal is to relate network architecture, described by the matrix $\bfW$, to coherence 
in network dynamics.  At second order, coherence is measured by the cross-covariance between the activities of nodes
$i$ and $j$ as a function of time lag $\tau$, 
$(\bfC_y(\tau))_{ij}$~\cite{corr_def}. As  computations are simpler in the spectral domain, we first consider the
cross-spectra,  
$\bfS_y(\om)=\EVs{\, \overline{{\bfyt}} \bfyt^T }$~\cite{ft_sp_def}, of the processes ($\tilde{x}(\om)=\mathcal{F}(x(t))$ represents the Fourier transform, $\bar{z}$ is a complex conjugate, $^T$ denotes a transpose, and bold symbols represent column vectors or matrices).
Cross-spectra and cross-covariances are related by the \WK{} Theorem, $\bfS_y(\om)=\mathcal{F}(\bfC_y(\tau))$~\cite{Laing:2009tw}.

After a Fourier transformation,  
the matrix form of Eq.~\eqref{e:lin} is 
\beq
\label{e:lin_freq}
\bfyt(\om)=\bfxt(\om)+\At(\om)\bfW \bfyt(\om).
\eeq
If the spectral radius $\Psi(\At(\om)\bfW)<1$, 
 then Eq.~\eqref{e:lin_freq} implies $\bfyt=(\bfI -\At\bfW)^{-1}\bfxt$, where $\bfI$ is the identity matrix. This leads to the following relation between
the matrix of cross-spectra, and  auto-spectra of the isolated (baseline) nodes,
\beq
\label{e:cov}
\bfS_y(\om)=(\bfI -\bar{\At}\bfW)^{-1}\bfS_x(\om)( \bfI-\At\bfW^T)^{-1}.
\eeq
This shows how the baseline variability within individual nodes, $\bfS_x(\om)=S_x (\om) \bfI$, propagates through the network.
An analog of Eq.~\eqref{e:cov} holds for networks of integrate--and--fire neurons and Hawkes processes~\cite{Hawkes:1971-2,Trousdale:2012}. 
%Our findings about pairwise covariances extend to these systems.

Eq.~\eqref{e:cov} can be expanded in a series~\cite{Trousdale:2012,Hu:2012vh,Pernice:2011},
\beq
\label{e:cov_epn}
\bfS_y(\om)/S_x(\om)=\sum_{n,m=0}^\infty \bar{\At}^{n}\At^{m}\bfW^n(\bfW^T)^m.
\eeq
The cross-spectra are normalized by $S_x(\omega)$ to obtain a  unitless 
measure of network coherence, which we can use to
approximate average correlation coefficient (see~\cite{Hu:2012vh}).

As shown by~\cite{Pernice:2011,Trousdale:2012}, the sum in Eq.~\eqref{e:cov_epn} represents contributions to the cross spectrum from paths (i.e., motifs) within the network.  Several  second order motifs are shown in Fig.~\ref{F:0}.  For instance, the second order term  $\bar{\At}\At(\bfW \bfW^T)_{ij}=w^2 |\At|^2\sum_{k}\bfW^0_{ik}\bfW^0_{jk}$ counts all contributions to the cross-spectrum of nodes $i$ and $j$ due to common input from nodes $k$ (the rightmost motif in Fig.~\ref{F:0}(b)). 
In general $(\bfW^n(\bfW^T)^m)_{ij}$ represents the contribution of  \emph{$(n,m)$ motifs} which consist of two directed chains of length $n$ and $m$ emanating from a single apex and terminating in nodes $i$ and $j$, respectively.   See Fig.~\ref{F:illustration}. The same node can be visited multiple times, and the $(0,m)$ motif is a  chain of length $m$.

\begin{figure}[H]
\centering
\includegraphics[scale=0.2]{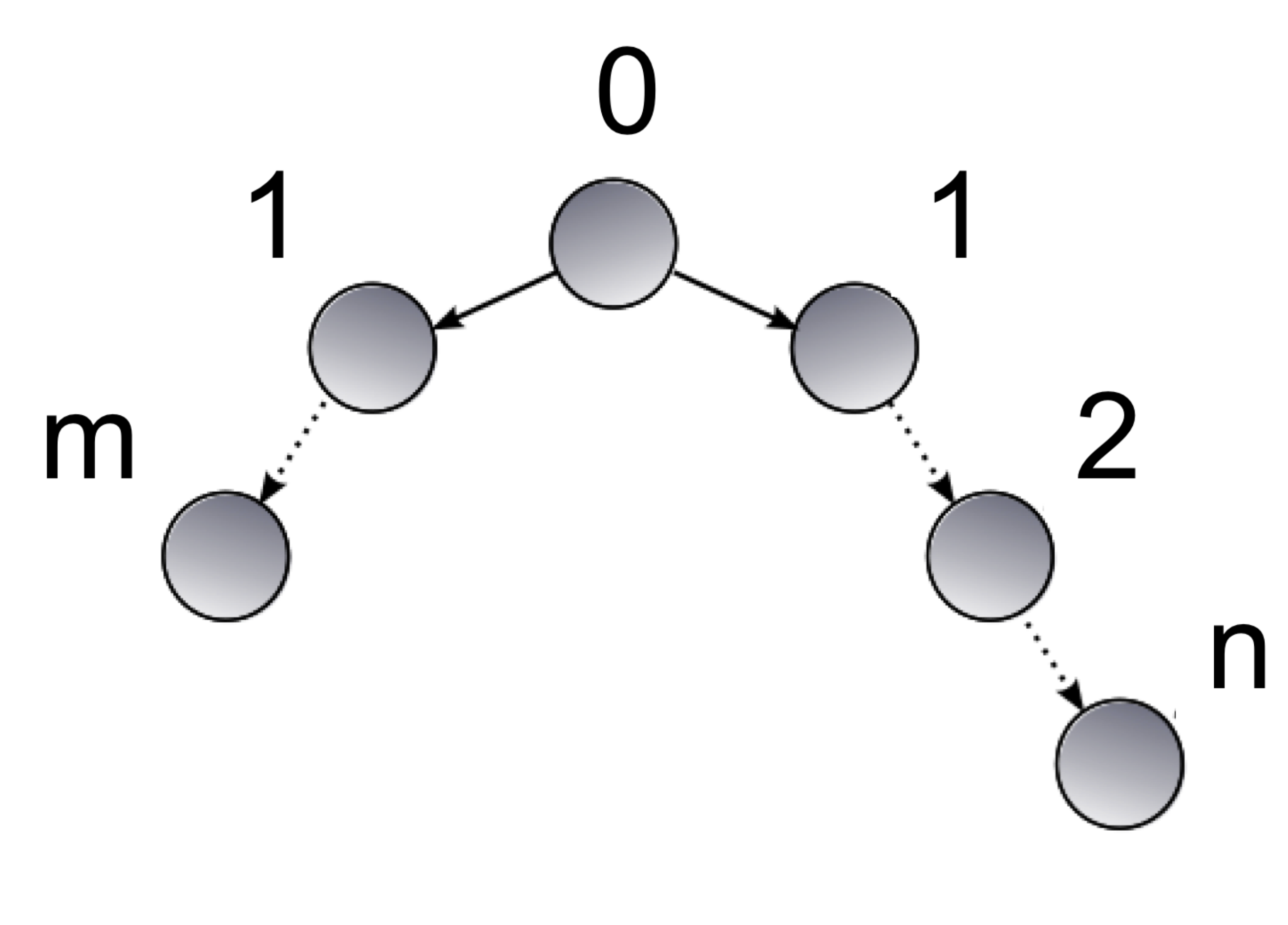}
\renewcommand{\figurename}{FIG.}
\caption{Illustration of a $(n,m)$ motif.}
\label{F:illustration}
\end{figure}

Fig.~\ref{F:1} illustrates such an expansion for two mutually inhibiting nodes (see also~\cite{Trousdale:2012}). The cross-covariance between the nodes is shown in Fig.~\ref{F:1}(a) with
contributions of low order motifs in Fig.~\ref{F:1}(b). As motif order increases, corresponding contributions to the cross-covariance decrease in magnitude, but increase in width. The asymmetry of a  contribution  increases with the asymmetry of the associated motif,
\emph{i.e.} the difference between  $n$ and $m$ in an $(n,m)$ motif: Compare the contributions of the $(1,2)$ and $(0,3)$ motifs. A graphical decomposition of the circuit into the first few $(n,m)$ motifs is shown in the inset of Fig.~\ref{F:1}(b).
Since the network is recurrent, the expansion in Eq.~\eqref{e:cov_epn} does not terminate as a node can appear multiple times in  a motif.

\begin{figure}[H]
\centering
\includegraphics[scale=1.4]{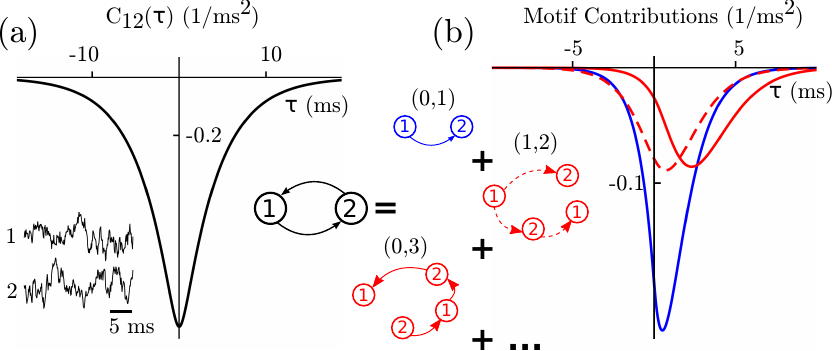}
\renewcommand{\figurename}{FIG.}
\caption{ (a) The cross-correlation function of two mutually inhibiting nodes modeled by an \OU{} process; inset shows traces of $y_1(t)$ and $y_2(t)$.  (b) Contribution of first- and third-order motifs to the cross-correlation function in  (a). (Middle Inset) Diagrammatic expansion of the network showing motifs whose contributions are given, via the same line and color types, in  (b).}
\label{F:1}
\end{figure}

\section{Moments, cumulants, and network-wide coherence}
\label{S:motif_cumulants}

\begin{figure}
\centering
{\label{f:motif_diagram}
\includegraphics[scale=1.4]{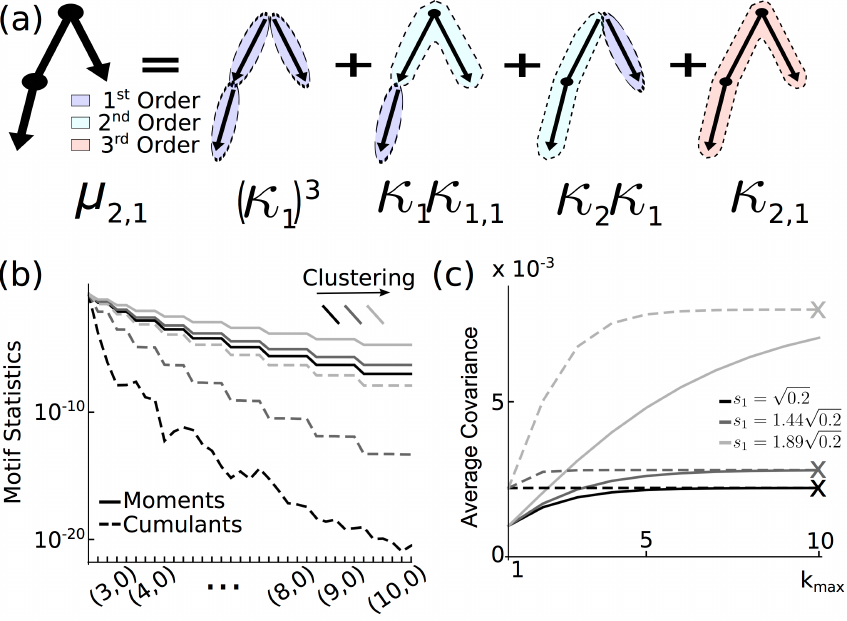}}\\
% switched to .pdf to make compiling faster, change back to .png if needed
% JT: .pdf compiles slightly faster, but gives me a big freeze as I scroll through, so I switched back to .png
\renewcommand{\figurename}{FIG.}
\caption{(a) The probability of observing the motif ($\M_{2,1}$) decomposed in terms of
\emph{motif cumulants} of the graph. (b) The magnitude of motif cumulants (dashed lines) and moments (solid lines) for stochastic block networks with $N=1000$, $p=0.2$. Clustering in network connectivity increases from darker to lighter lines (see text for definition of clustering and network structure; precise values as specified in panel (c)).  The $(n,m)$ motifs with $n\ge m$ are listed first by order $n+m$ then arranged by increasing $n$ within the motifs of same order. (c) Approximations of average covariances using motif moments (Eq.~\eqref{e:cov_epn_avg}) and cumulants (Eq.~\eqref{e:rsm}) truncated at order $k_{\max}$. Crosses indicate exact values obtained from Eq.~\eqref{e:cov}.}
\label{f:figure2}
\end{figure}

We next relate  network coherence and network structure using motif statistics.
For concreteness -- but without loss of generality \cite{Pernice:2011,Trousdale:2012} -- we consider the total covariance between pairs of nodes. This is equivalent to evaluating
all spectral quantities  at $\om = 0$, and we indicate this by suppressing dependences on $\om$.  
We measure network-wide coherence using the average of this total covariance over all pairs of nodes. 
As in~\cite{Pernice:2011,Hu:2012vh,Trousdale:2012}, if we denote by $\langle {\bfX} \rangle$ the empirical average of the entries of matrix $\bfX$, we obtain from Eq.~\eqref{e:cov_epn}
\beq
\label{e:cov_epn_avg}
\begin{split}
\langle \bfS_y \rangle /S_x
&= \sum_{n,m=0}^\infty \At^{n+m} \langle \bfW^n(\bfW^T)^m \rangle  
=: \frac{1}{N}\sum_{n,m=0}^\infty  g^{n+m}\M_{n,m}, 
\qquad \text{where} \qquad g=N\At w.
\end{split}
\eeq
Here the \emph{motif moment}, $\M_{n,m} = \langle \bfW^{0n}(\bfW^{0T})^m\rangle/N^{n+m-1}$, is the empirical probability of observing an $(n,m)$ motif in the network~\cite{Pernice:2011,Hu:2012vh}.  Note that the empirical average 
is defined over a particular realization of the adjacency matrix $\bfW^0$.   
We define  $\M_{n,0} =\M_{n},$ and  let $\M_{0,0}=1$. The entire hierarchy of motif moments, $\M_{n,m}$, needs to be known to  evaluate Eq.~\eqref{e:cov_epn_avg} exactly. In practice, only a subset of $\M_{n,m}$, up to a certain order $n+m\leq k_{\max}$, is known and can be used with Eq.~\eqref{e:cov_epn_avg} to approximate network-wide covariance.

Truncating Eq.~\eqref{e:cov_epn_avg} at some order yields an approximation of average 
coherence in terms of 
motif moments up to that order. However, these approximations
can exhibit significant deviations from the true value~\cite{Hu:2012vh}.  
Previously, we introduced an alternative, `motif resumming approximation'~\cite{Hu:2012vh}, which provided a series expansion of average coherence in terms of motif \emph{cumulants} (defined below) rather than
motif moments.  
Truncation of the resulting series yielded a significantly improved approximation of average coherence, given the same set of motif frequency data.

While we earlier provided a probabilistic interpretation of this motif cumulant approach, a general framework was missing~\cite{Hu:2012vh}.  We next provide such a framework, by reexamining the motif cumulants $\K_{n,m}$ that first appeared in~\cite{Hu:2012vh}. We provide a novel definition which clarifies the underlying combinatorial relationship between motif cumulants $\K_{n,m}$ and motif moments $\M_{n,m}$,  analogous to that between cumulants and moments of a random variable.  Equipped with this new definition, we are able to express dynamical correlations of all orders in terms of motif cumulants (Sec.~\ref{S:hoc}).

The construction of motif moments from cumulants is based on a familiar interpretation: estimating  the probability of a joint event from the probability of its constituents.    Fig.~\ref{f:figure2}(a) demonstrates this for an example motif.  Each term in the decomposition of this (2,1)  diverging motif arises from a cumulant of smaller or equal order.  The first term corresponds to the probability of the motif occurring in a network with edges chosen independently, \emph{i.e.}
an \ER{} network. Subsequent terms give corrections from excess occurrences of second and third order submotifs.
Thus, each motif cumulant, $\K_{n,m}$, captures ``pure" higher order connectivity statistics. Such decomposition can also be expressed in combinatorial form. Let $\mathcal{C}(n)$
be the set of all compositions (\emph{ordered} partitions) of $n$.
Then
\begin{eqnarray}
\label{E:recursive1}
\M_{n} &=&
\sum_{\{n_1,  \cdots , n_t\} \in \mathcal{C}(n)} \left(\prod_{i=1}^{t} \K_{n_i}\right) \\
\label{E:recursive2}
\M_{n,m} &=&
\sum_{\substack{\{n_1,  \cdots , n_t\}  \in \mathcal{C}(n) \\ \{m_1,  \cdots , m_s\} \in \mathcal{C}(m)}}  \left(\prod_{i=2}^{t} \K_{n_i}\right) \left(\K_{n_1,m_1} +\K_{n_1}\K_{m_1}\right) \left( \prod_{j=2}^{s} \K_{m_j} \right)
\end{eqnarray}
In evaluating these terms, we set $\left(\prod_{i=2}^{t} \K_{n_i}\right)=1$ if $t=1$.  

Expressions~\eqref{E:recursive1}-\eqref{E:recursive2} define the full set of $\K_{n,m}$ recursively.  These are related directly to coherent network dynamics in the theorem that follows.

\begin{theorem}
\label{TH:cumulant}
For a network with dynamics defined by Eq.~\eqref{e:lin}, with $\bfW=w\bfW^0$, the mean coherence (i.e., averaged across all cell pairs) is given by
\beq
\label{e:rsm}
\frac{\langle\bfS_y \rangle}{S_x}
=\frac{1}{N} \left(1 - \sum_{n=1}^\infty g^n \K_n  \right)^{-2} \left(1 +  \sum_{n,m=1}^	 \infty  g^{n+m}  \K_{n,m}\right) \;,
\eeq
where $g=N\At w$.
The  $\K_{n,m}$ here  are the motif cumulants, defined by ~\eqref{E:recursive1}-\eqref{E:recursive2}.
\end{theorem}
To prove this result we demonstrate a relation between the cumulants, $\K_{n,m},$ and the quantities expressed in terms of matrix products in Eq.~(32) of~\cite{Hu:2012vh}.  Eqn.~\eqref{e:rsm} then follows immediately from substituting the $\K_{n,m}$ into Eq.~(32) of~\cite{Hu:2012vh}.  The proof is given in Appendix~\ref{S:exp_motif}.

\bigskip

In Fig.~\ref{f:figure2}(c) we compare the expressions for network coherence in terms of motif moments (Eq.~\eqref{e:cov_epn_avg}) and motif cumulants (Eq.~\eqref{e:rsm}).  We compute both expansions for three example networks (whose construction and differences will be the topic of later sections); for each, we  illustrate how motifs of increasing order contribute to predicted network coherence. 

This illustrates a general phenomenon. Truncating Eq.~\eqref{e:cov_epn_avg}, and keeping only terms with  $n+m \leq k_{\max}$, approximates the contributions of these $(n,m)$ motifs to the mean dynamical coherence in the network.
 A similar truncation of Eq.~\eqref{e:rsm} however approximates coherence
in terms of contributions of paths of {\it all} orders. In this latter case,  frequencies of motifs of order exceeding $k_{max}$ are \emph{predicted} from the observed frequencies of motifs of
order up to $k_{max}$. 
Fig.~\ref{f:figure2}(c) shows that these predictions are useful: values of correlations based on cumulants converge more quickly than those derived from motif moments. The difference can be explained by looking at the magnitude of the cumulants/moments against the order (Fig.~\ref{f:figure2}(b)). Importantly, cumulants decay much faster than moments in all three cases --- hence the increased accuracy of Eq.~\eqref{e:rsm} over Eq.~\eqref{e:cov_epn_avg} at a given order.

Fig.~\ref{f:figure2}(b) also illustrates that heterogeneity in network architecture can impact how quickly cumulants and 
moments decay, an observation we will revisit. The networks used in Fig.~\ref{f:figure2}(b) and (c) have a variable degree of \emph{clustering} or  ``clumping" in network connectivity ---  we precisely define our graph generation rules below. A greater degree of
clustering results in a slower decay of both motif moments and cumulants. Higher order statistics are necessary to accurately describe the structure of such networks.  Hence, with more heterogeneity in connections across a network, the frequency of larger, more complex graph motifs has a greater impact on network coherence.

\section{Heterogeneous networks and subpopulation cumulants} 
\label{S:heter_sub}

Motif cumulants --- via Eq.~\eqref{e:rsm} --- provide a way to estimate global dynamical correlation in terms of local network structure.  As illustrated above, the accuracy of such approximations depends on the network's architecture (see Fig.~\ref{e:cov}).  We next highlight the key impact of heterogeneity or clustering in network connectivity  on the approximation.  We then introduce a partitioning approach, and the allied concept of subpopulation cumulants, which allow us to relate local network structure to dynamics even in  heterogeneous networks.

\subsection{Heterogeneity in network architecture}
\label{S:heter_network}

To study the impact of heterogeneity on the approximation given by the motif cumulant method, we first consider the \emph{stochastic block network model}~\cite{Wang:1987tc,Daudin:2008,LitwinKumar:2012ue} illustrated in Fig.~\ref{F:0}. Such networks are comprised of two subpopulations (or clusters) of size $N/2$ (indicated by circular and triangular nodes). Each cluster is associated with a constant $s_i$, $i=1,2$, and the connection probability between nodes in subpopulation $i$ and $j$ is $p_{ij}=s_is_j$.
With fixed overall connection probability $p$, the difference between $s_1$ and $s_2$  describes the degree of clustering in the network.  The case $s_1 = s_2$ corresponds to an \ER{} network (no clustering), while $s_1=2 \sqrt p$, $s_2=0$ implies that only nodes in the first subpopulation are connected (extremal clustering).

To illustrate the impact of clustering we generate three networks with different values of $s_1,s_2$ in Fig.~\ref{f:figure2}(b,c).  Comparing pairs of curves (moments and cumulants) with different shades (i.e., different degrees of clustering) reveals the dependence of motif moments and cumulants on graph structure. The magnitude of motif moments and cumulants of a given order increases with clustering (Fig.~\ref{f:figure2}(b)). Hence in clustered, heterogeneous networks large motifs can strongly impact dynamical coherence (Fig.~\ref{f:figure2}(c)). Moreover, network motifs of increasing order are needed to accurately predict dynamical correlations as clustering increases.

As a more complex example, we also considered the \BA{} model. We find that the  behavior of the two models is
similar (Fig.~\ref{f:figure2_SF} in the Appendix). Such similarity is consistent with observations reported in the literature~\cite{Prettejohn:2011tu} and underscores the generality of the impact of network heterogeneity.

%Networks generated using the \BA{}  model behave similarly  \cite{Prettejohn:2011tu}, a point we return to below (See also Fig.~\ref{f:figure2_SF} in the Appendix).

These results agree with intuition. \ER{} networks have an architecture that is ``statistically homogeneous," as the probability of each link occurring in the network is the same.  Thus, the most local network statistic -- connection probability -- fully determines graph structure and hence the level of dynamical coherence.
Similarly, `nearly \ER{}' networks are without significant graphical heterogeneity, and low order motif cumulants can accurately predict dynamical coherence.    On the other hand, in highly clustered networks the probability of a path between a set of nodes depends on higher order connectivity statistics.  As a result, the frequency of large motifs cannot be obtained accurately from the frequencies of smaller ones.  In such networks higher-order motif statistics have a significant impact on dynamical coherence.

The necessity of estimating the frequency of higher order motifs could limit the applicability of this approach.  In many situations the full connectivity structure of a network is not known, and global properties of the network are difficult to estimate. For instance, in the case of biological neuronal networks, the number of neurons which can be simultaneously recorded
in order to map out their connectivity is often limited to only a small handful~\cite{Song:2005, Perin:2011}.  Moreover, many networks possess additional structure past the simple heterogeneities discussed above  -- for instance, neuronal networks may be composed of both excitatory and inhibitory cells. Accounting for such natural
subdivisions of the graph can lead to more accurate approximations of dynamical coherence.

\subsection{Subpopulation cumulants}

We next show how to subdivide a network  to tame the effects of heterogeneity in  architecture, and re-establish the link between local connectivity and global coherence. Subsets of nodes in graphs can be grouped into classes, or subpopulations, that share features of dynamics or
connectivity.  Once a division is given, we can characterize each subpopulation by its own motif statistics. These subpopulation motifs are first introduced in~\cite{Hu:2012vh} in the context of studying neural networks with two different types of cells. However, a key difference here is that division or grouping of nodes may not be given in advance, but can be obtained (as we will show) from the network architecture. How the nodes are subdivided can affect the accuracy of the motif cumulant method, a matter we will address in the next section. First, we extend the ideas in~\cite{Hu:2012vh} to the general case of $b$ populations using the new combinatorial definition of motif cumulants introduced in Sec.~\ref{S:motif_cumulants}.

For $b$ subpopulations,   $\M_{n,m}$ becomes a $b\times b$ matrix of motif moments.  Entry $p,q$ of this matrix is the empirical probability
of  an $(n,m)$ motif with end nodes belonging to populations $p$ and $q$, respectively.  Let $V$ be the set of all nodes, and $V_\alpha, \alpha=1,\cdots,b$ be the set of nodes in  population $\alpha$. We denote the size of each population by $N_\alpha=|V_i|$. We then have
\begin{eqnarray}
(\M_{n,m})_{p,q}
\nonumber
&=&\sum_{i_n,i_{n-1},\cdots, i_0=j_0, j_1,\cdots,j_{m-1}, j_{m}}
\bfW^0_{i_n,i_{n-1}} \bfW^0_{i_{n-1},i_{n-2}} \cdots \bfW^0_{i_{1},i_0}
\bfW^{0}_{j_1,i_0}\cdots \bfW^{0}_{j_{n},j_{n-1}} /Z \\
\nonumber
&=&
\sum_{i_n,i_{n-1},\cdots, i_0=j_0, j_1,\cdots,j_{m-1}, j_{m}}
\bfW^0_{i_n,i_{n-1}} \bfW^0_{i_{n-1},i_{n-2}} \cdots \bfW^0_{i_{1},i_0}
\bfW^{0T}_{i_0,j_1}\cdots \bfW^{0T}_{j_{n-1},j_n}/Z\\
&=&(\langle  (\bfW^0)^{n} (\bfW^{0,T})^m \rangle_{B})_{p,q}/N^{n+m-1}
\end{eqnarray}
In these sums we assumed that the indices satisfy $i_n \in V_p$, $j_m \in V_q$, and other $i_{s<n}$, $j_{t<m}$ are chosen from $V$.
We also used the normalization factor $Z=N^{n+m-1}N_p N_q$, while $\langle  \bfM \rangle_{B}$ represents the block average of a matrix according to the division of populations, i.e. $(\langle  \bfM \rangle_{B})_{p,q}=\frac{1}{N_p N_q}\sum_{i\in V_p,j \in V_q}M_{i,j}$.

This partition of nodes and motifs into subpopulations is depicted in Fig.~\ref{F:0}, where the color of a node indicates its class.
Motifs may involve either nodes of a single class, or a combination of the two. 

Motif cumulants $\K_{n,m}$ are $b\times b$ matrices that are defined by recursive relationships similar to Eqs.~(\ref{E:recursive1},\ref{E:recursive2}):
\begin{eqnarray}
\label{E:recursive1_sub}
\M_{n} &=&
 \sum_{\substack{\{ n_1,\ldots,n_t\} \in \mathcal{C}(n)} }\left[ \left(\prod_{i=2}^{t} \K_{n_i} \bfE \right) \K_{n_1} \right] \\
\label{E:recursive2_sub}
\M_{n,m}&=&
 \sum_{\substack{\{ n_1,\ldots,n_t\} \in \mathcal{C}(n) \\ \{ m_1,\ldots,m_s\} \in \mathcal{C}(m)}} \left(\prod_{i=2}^{t} \K_{n_i} \bfE \right) \left(\K_{n_1,m_1} +\K_{n_1} \bfE \K_{m_1}^T\right) \left( \prod_{j=2}^{s} \bfE \K_{m_j}^T \right).
\end{eqnarray}
Here $\bfE=\mathrm{diag}\{N_1/N,\cdots,N_b/N \}$ is inserted between each motif cumulant matrix multiplication and yields the appropriate weighted sums for the interpretation of the terms $\M_{n,m}$ and $\K_{n,m}$ as \emph{probabilities}. Specifically, scaling by $\bfE$ is multiplication by the probability of selecting nodes from respective populations at ``breaks" in the motifs.

How should these population-specific motif cumulants be combined to estimate the average correlation? An extension of Eq.~\eqref{e:rsm} was developed for two populations in~\cite{Hu:2012vh}, and stated in terms of matrix products.  This generalizes immediately to the case of an arbitrary number of populations, $b$, and -- as in Theorem~\ref{TH:cumulant} above -- can be restated in terms of (matrix-valued) motif cumulants.  The result is:

\begin{corollary}
\label{TH:cumulant_multi}
Let $\langle \bfS_{y} \rangle_B$ represent a block-wise average over entries corresponding to each subpopulation. For a network with dynamics defined by Eq.~\eqref{e:lin} with $\bfW=w\bfW^0$, the generalization of Eq.~(\ref{e:rsm}) to subpopulation motif cumulants is~\cite{Hu:2012vh}:
\beq\label{E:rsm_sub}
\begin{split}
\langle \bfS_{y} \rangle_B/ S_x&:=\bfD\bfU^T\bfS_{y}\bfU\bfD/S_x\\
&=\frac{1}{N} \left(\bfI - \sum_{n=1}^\infty g^n \K_n \bfE\right)^{-1}\left(\bfE^{-1}+  \sum_{n,m=1}^\infty  g^{n+m} \K_{n,m} \right)
\left(\bfI - \sum_{m=1}^\infty g^m \bfE \K_m^T \right)^{-1}.
\end{split}
\eeq
where $\bfD=\mathrm{diag}\{1/\sqrt{N_1},\cdots,1/\sqrt{N_b} \}$, $\bfU$ is $N\times k$ matrix given by $\bfU=[\bfu_1|\cdots|\bfu_k]$, and $\bfu_i=(0,\cdots,0,1,\cdots,1,0\cdots,0)^T/\sqrt{N_i}$ is the vector where the nonzero entries appear only at indices that match one of the nodes in the given subpopulation, normalized to unit $L_2$ norm.
The  $\K_{n,m}$ here  are the subpopulation motif cumulants, defined in Eqs.~(\ref{E:recursive1_sub}--\ref{E:recursive2_sub}).
\end{corollary}

\noindent The arguments necessary to establish this Corollary are given in Appendix~\ref{S:sub_pop_rsm}.

%As is evident,  
% and reproduced in Appendix~\ref{S:sub_pop_rsm} for completeness,  gives average covariances in terms of the matrices $\K_{n,m}$ (Eq.~\eqref{E:rsm_sub}). In fact, the two formulae (and the derivations) are very similar once scalar quantities (whole network motif cumulants) are replaced by matrices (consisting of motif cumulants within and across subpopulations).

In Fig.~\ref{f:figure3}(a) we use stochastic block model networks to demonstrate the subpopulation motif approach. The structure of such networks is defined using two groups of nodes with different connectivity.  We group nodes accordingly into two populations and apply the subpopulation cumulant formula given by Eq.~\eqref{E:rsm_sub}. The resulting approximation of average correlations is a significant improvement over that obtained using a single population: First order motif cumulants {\it alone} perfectly predict average correlations; whereas we require motifs of order up to 4 or 5 orders for the same networks if we use a single population approach (Fig.~\ref{f:figure2}(b)).

%Fig.~\ref{f:figure3}(a) demonstrates the subpopulation motif approach with the stochastic block model networks. Here we group nodes according to the natural subpopulations in the stochastic block network model. We see that first order motif cumulants {\it alone} perfectly predict average correlations (compare with Fig.~\ref{f:figure2}(b)).

Importantly, the subpopulation approach also works when there is no obvious way to group the nodes.  As an example, consider the highly heterogeneous \BA{} networks.   
If we order nodes by degree,  two subpopulations can be formed from nodes with degrees above and below a given threshold.  Fig.~\ref{f:figure3}(b) shows that this approach substantially simplifies the link between network structure and dynamics:  if the subpopulations are chosen optimally, covariance in the network dynamics can be accurately predicted using motifs of only order two, while motifs up to order four or five are needed otherwise.

\begin{figure}[H]
\centering

{\label{f:sub_HE}
\includegraphics[scale=1.4]{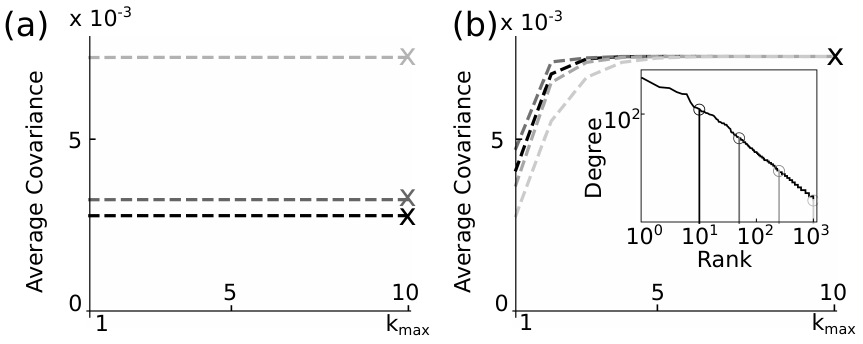}}
\renewcommand{\figurename}{FIG.}
\caption{Approximations of average covariances using the subpopulation cumulant approach, truncating at order $k_{\max}$.  Crosses indicate exact values obtained from Eq.~\eqref{e:cov}. (a) Stochastic block model networks of Fig.~\ref{f:figure2} (same color scheme) divided into two subpopulations 
--- first order motif cumulants now provide a complete description of
the network structure; (b) \BA{} network divided into two subpopulations according to whether the sum of in- and out-degrees of each node
lie above or below different thresholds (inset shows cut-off degree ranks (descending)).}
\label{f:figure3}
\end{figure}

\subsection{How to partition a network and why it works} 
\label{S:partition}

In \cite{Hu:2012vh}, we provided an intuitive explanation of why motif cumulants provide a better approximation of network-wide covariance (Eq.~\eqref{e:rsm}) than  motif moments  (Eq.~\eqref{e:cov_epn_avg}). In this section we extend this argument to heterogeneous architectures. In doing so, we will reveal why network partitioning can work so well, describe a rule of thumb and apply it to a general network.

First, we review the arguments in~\cite{Hu:2012vh} for statistically homogeneous (e.g. \ER{}) networks.
The argument was based on studying the spectral radii $\Psi(\bfW^0)$ and $\Psi(\bfW^0 \bfTheta)$,  where $\bfTheta=\bfI-\bfu \bfu^T$ and $\bfu=(1,\cdots,1)^T/\sqrt{N}$. Using the matrix expression of motif statistics (see Appendix Eq.~(\ref{E:explicit1}-\ref{E:mn_explicit})), it is straightforward to see that those spectral radii are related to the asymptotic rate of decay of the moments, $\M_{n,m}$,~\cite{Pernice:2011} and cumulants, $\K_{n,m},$ respectively. 

The faster decay of  cumulants compared to moments is therefore reflected by $\Psi(\bfW^0\bfTheta)$ being much smaller than $\Psi(\bfW^0)$. This is indeed the case for networks with sufficiently ``homogeneous" connectivity~\cite{Hu:2012vh}, cf.~\cite{Pernice:2011}: For \ER{} networks, the spectrum of $\bfW^0$ is characterized by a bulk part with many eigenvalues distributed over a region near 0 in the complex plane, and one single positive eigenvalue with much larger magnitude. This latter eigenvalue determines $\Psi(\bfW^0)$  (from the Perron-Frobenius theorem~\cite{Horn:1990}, cf.~\cite{Rajan:2006}), and therefore the rate of decay of the moments $\M_{n,m}$. 
To study $\Psi(\bfW^0\bfTheta)$, and therefore the rate of decay of the motif cumulants, we first define the ``PF vector" as the eigenvector associated with the outlying eigenvalue of $\bfW^0$ in an arbitrary network. For sufficiently ``homogeneous" networks such as \ER{} networks, the PF vector is close to $\bfu$ as a reflection of the underlying homogeneity. 
%\Kcomment{This could be a bit confusing.  I think the PF vector is defined for all networks.  But here we say ``As a reflection of the underlying homogeneity,''.  What if the network is heterogeneous? In particular, we should make it clear that we first consider nearly homogeneous networks, and then heterogeneous networks.  Perhaps split this into two paragraphs.} \Ycomment{This paragraph should be all about homogeneous networks. There is a gap of narration caused by defining the PF vector (for all networks). I also rewrite the last sentence above, hopefully clear now. }
Note that multiplication by $\bfTheta$ essentially removes the eigenvalue associated to this vector from the spectrum of $\bfW^0\bfTheta$, since  $\bfW^0\bfTheta \bfu=\bfW^0 0=0$. This leads to the significant reduction of $\Psi(\bfW^0\bfTheta)$ compared to $\Psi(\bfW^0)$.

To extend such intuition to heterogeneous networks, we need to answer two questions:  First, what is the PF vector for heterogeneous networks? Second, how does dividing a network into subpopulations change the counterpart of $\bfTheta$, and the resulting spectrum?

We first observe that for many networks, the PF vector is approximately the (in) degree list, denoted by $\bfd$ (normalized to unit  $L_2$-norm). In particular, we have found numerically that this is the case for stochastic block models and the \BA{} networks we consider (see Fig.~\ref{F:PF_vec} in Appendix). We will use this observation about the PF vector in making intuitive arguments below, but first pause to make some general, heuristic comments as to its possible justification. We begin by referring back to the case of \ER{} networks, where the PF vector approximately proportional to the homogeneous vector $\bf u$ as stated above; and for large matrices, $\bf u$ will also be approximately proportional to the degree vector $\bf d$ with small (relative) error. Now looking at the ensemble average $\EVb{\bfW^0}\propto \bfu \bfu^T$ and observe that $\bfu$ is the (exact) PF vector for this average matrix $\bfu \bfu^T$. Thus the PF vector for the ensemble average and for realizations of the adjacency matrices agree --- although this relies on the probabilistic structure of the underlying random matrices in a much more complicated way than we attempt to describe. Next, for a more general graph model, consider an adjacency matrix with an ensemble average that can be written in rank-one form: $\EVb{\bfW^0}=\bfa \bfb^T$ (where $\bfa,\bfb$ are column vectors with nonnegative entries). The PF vector for $\EVb{\bfW^0}$ is $\bfa$; moreover, this is once again proportional to the (average) in-degree list. An analogy with the \ER{} case suggests a possible reason for why the PF vector for individual adjacency matrices $\bfW^0$ are also found to be approximately proportional to $\bf d$ --- although this argument is not rigorous.

We now discuss how to use the fact that the PF vector $\propto \bf d$ to best partition a network into subpopulations.  Recall that the subpopulation theory can be viewed as formally substituting the scalar motif moment and cumulant quantities in the original theory with $b \times b$ matrices (Eq.~(\ref{E:recursive1_sub}-\ref{E:rsm_sub})). In \cite{Hu:2012vh}, we showed that the matrix expression for $\K_{n,m}$ and $\M_{n,m}$ are given by Eqs.~(\ref{E:explicit1}-\ref{E:mn_explicit}), where repetitive factors such as $\bfW^0\bfTheta_B$ appear in places of $\bfW^0\bfTheta$. Here $\bfTheta_B$ is a block diagonal generalization of $\bfTheta$ for the subpopulation approach.  In particular,
\beq
\bfTheta_{B}=
\left(
\begin{array}{ccc}
\bfTheta_1 & &\\
& \ddots &\\
& & \bfTheta_b
\end{array}
\right),
\eeq
where each diagonal block corresponds to a subpopulation.  Here $\bfTheta_i=\bfI_{N_i}-\bfu_{N_i} \bfu_{N_i}^T$ (where $\bfu_{N_i}=(1,\cdots,1)^T/\sqrt{N_i}$) is an ``original" $\bfTheta$ matrix, simply defined with population size $N_i$. 

Combining the above observations, we look for a partition of the network that will bring $\bfW^0\bfTheta_B \bfd$ as close to 0 as possible. First, consider the stochastic block model. Note that the $\bfTheta_B$ we defined above will map to $\bf 0$ any vector that is piecewise constant over the indices of each subpopulation.  Therefore, if we choose the network partition naturally provided by the stochastic blocks themselves, we obtain $\bfW^0\bfTheta_B \bfd=0$.  As expected, this partitioning results in very rapid decay of motif cumulants, and hence an ability to predict network coherence using only low order motif statistics (here, order 1; see Fig.~\ref{f:figure3}(a)). 

For the \BA{} network, there are no ``natural" subpopulations, but partitioning still leads to a significant improvement in predictions of network coherence.  In this case, continue to divide the network into just two subpopulations (Fig.~\ref{f:figure3} (b)).  The goal is to perform this division in so that it will minimize $\Vert \bfW^0\bfTheta_B \bfd \Vert_2$.  In practice, we instead consider the simpler question of minimizing  ${\lVert \bfTheta_B \bfd \rVert_2}$ as an approximation.  As noted above, $\lVert \bfTheta_B \bfd \rVert_2$ measures the error of a piecewise constant (over the indices of subpopulations) approximation of $\bfd$.  In Fig.~\ref{f:sub_pop_v}, we plot this error against a threshold parameter in node degree that is chosen to partition the network; this shows that the error is minimized at a cut-off degree ranking of roughly 30-40 (across different random realizations of a \BA{} network with the same parameters).  As expected from our heuristic arguments, this value is close to the value of the threshold that gave the most rapid convergence of the cumulant-based estimates of network covariance (degree ranking = 50, Fig.~\ref{f:sub_pop_v}).

\begin{figure}[H]
\centering
\includegraphics[scale=1.4]{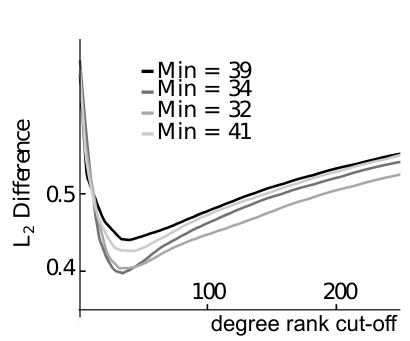}
\renewcommand{\figurename}{FIG.}
\caption{
$L_2$ norm of the difference between the degree list (normalized) and the piecewise constant vector (see text) given by a certain cut-off ranking of the degrees. Different shades are 4 realizations of \BA{} networks (with same parameters). The legend is the cut-off degree ranking that achieves the minimum of difference.}
\label{f:sub_pop_v}
\end{figure}

Up to this point we have defined motif cumulants, and shown how they can be used to make accurate predictions of coherence in average network activity.   These were results about second-order correlations (i.e., covariances)  averaged across node pairs.  We next extend the theory of motif cumulants to correlations of arbitrary order.

\section{Higher order correlations}
\label{S:hoc}
Here we show how to generalize our theory to relate higher order statistics of a network's dynamics to its architecture.  While the second-order results above can be used for both finite-valued stochastic systems (i.e., OU and jump processes) and coupled point processes, the higher-order results are \emph{only valid in their present form for finite-valued stochastic systems} (not point processes with delta function pulses).   Extensions to higher-order coherence for {interacting point processes} are  nontrivial and will be tackled elsewhere. 

The $k^{th}$ order 
cross-covariance function for the processes  in  Eq.~\eqref{e:lin} are  defined using joint cumulants of random variables,
\beq
\label{E:C_tensor_def}
\bfC_{y[k]}^{i_1i_2\cdots i_k}(\tau_1,\ldots,\tau_{k-1}) :=
\kappa \left(y_{i_1}(t), y_{i_2}(t+\tau_1), \cdots, y_{i_k}(t+\tau_{k-1}) \right).
\eeq
A generalization of the \WK{} theorem relates the Fourier transform of the higher order cumulant to the polyspectra $\bfS_{y[k]}^{i_1i_2\cdots i_k}$~\cite{Brillinger:1964ub} defined via the Fourier transform of the processes. 
\beq
\label{E:WK_hoc}
\begin{split}
\mathcal{F} (\bfC_{y[k]}^{i_1i_2\cdots i_k})
&=\bfS_{y[k]}^{i_1i_2\cdots i_k}
=\kappa \left(\bar{\yt}_{i_1}(\om_1+,\cdots,+\om_{k-1}), \yt_{i_2}(\om_1), \cdots, \yt_{i_k}(\om_{k-1}) \right)\\
&:=\lim_{T\rightarrow \infty }\frac{1}{T}
\sum_{\chi} (|\chi|-1)!(-1)^{|\chi|-1} \prod_{B\in \chi} \EVb{\prod_{j\in B} \yt_{i_j}(\om_{j-1})} \delta\left(\sum_{j\in B} \om_j\right)
\end{split}
\eeq
Here $\om_0=-\sum_{j=1}^{k-1} \om_j$, $\yt_{i_j}(\om)=\int_0^{T} dt e^{-2\pi i \om t}(y_{i_j}(t) -
  \mathbf{E}\left[y_{i_j}(t)\right])$, $\delta(z)=1$ when $z=0$ and $\delta(z)=0$ otherwise. The first sum is over all partitions $\chi$ of set $\{1,\cdots,k\}$, and $B$, as an element of $\chi$, is a subset of $\{1,\cdots,k\}$, $|\chi|$ is the number of partitions in $\chi$. To illustrate this formula, we first note that, at third order, it reduces exactly to the ``bispectrum"~\cite{Kim:1979,Huber:1971,Brillinger:1964ub}
\[
\mathcal{F} (\bfC_{y[3]}^{i_1 i_2 i_3}(\tau_1,\tau_2))
=\bfS_{y[3]}^{i_1 i_2i_3}(\om_1,\om_{2}):=\EVs{\bar{\yt}_{i_1}(\om_1+ \om_{2})\yt_{i_2}(\om_1) \yt_{i_3}(\om_{2})}.
\]

It is easy to see that Eq.~\eqref{E:WK_hoc} is multilinear in the variables $\yt_{i_j}$. Using Eq.~\eqref{e:lin_freq}, we can therefore generalize Eq.~\eqref{e:cov} to obtain the polyspectra of the processes $\bfy$ in terms of that for $\bfx$ via the propagation matrix $\bfPt = (\bfI -\At\bfW)^{-1}$: 
\beq
\label{e:hoc}
\bfS_{y[k]}^{i_1 i_2 \cdots i_k}(\om_1,\cdots, \om_{k-1})
=\sum_{j_1,\cdots, j_k}\tilde{\mathbf{P}}_{i_1j_1} (\om_0) 
\bfPt_{i_2j_{2}}(\om_1)\cdots\bfPt_{i_k j_k}(\om_{k-1})
\bfS_{x[k]}^{j_1 j_2 \cdots j_k}(\om_1,\cdots,\om_{k-1}).
\eeq
For example, replacing Gaussian white noise which appeared in the {OU} process with ``{P}oisson kicks", \emph{i.e.} considering a shot noise process, yields non-zero $\mathbf{S}_{x[3]}$.

Next, expanding $\bfPt=\sum_{n=0}^\infty (\At \bfW)^n$ in Eq.~\eqref{e:hoc} leads to an expression for polyspectra analogous to Eq.~\eqref{e:cov_epn_avg}: 

\beq
\label{e:cov_epn_hoc}
\langle \bfS_{y[k]}\rangle / S_{x[k]} =   \frac{1}{N^{k-1}}\sum_{n_1,\cdots,n_k=0}^\infty  g^{|n|}\M_{n_1,\cdots,n_k},
\eeq
where $\vert n \vert=\sum_{i=1}^k n_i$ and $g=N\At w$ as defined in Eq.~\eqref{e:cov_epn_avg}. The motif moments $\M_{n_1,\cdots,n_k} =\sum_{i_1,\cdots,i_k,j}(\bfW^{0n_1})_{i_1 j} \cdots (\bfW^{0n_k})_{i_k j}/N^{|n|+1}$. For simplicity, in the formula above we again set $\om_1=\cdots=\om_{k-1}=0$, and assume homogeneous dynamics for each node. Here, $\bfS_{x[k]}^{i_1 i_2 \cdots i_k}=S_{x[k]} \delta^{i_1 i_2 \cdots i_k}$ is a diagonal tensor, since the $\{x_i\}$ comprise an uncoupled and uncorrelated network. 

The most interesting aspect of Eq.~\eqref{e:cov_epn_hoc} are the motif moments $\M_{n_1,\cdots,n_k}$.  For dynamical coherence (and hence polyspectra) of order $k$, these motif moments are the frequencies of $k$-branch motifs with $n_j$ nodes on each branch.  Fig.~\ref{f:motif_3}(a) illustrates such a motif $\M_{1,1,1}$, for $k=3$ branches and $n_j \equiv 1$ node on each branch.  Importantly, these $k-$branch motifs are the only ones that appear at each order in the series of Eq.~\eqref{e:cov_epn_hoc}.

We note that higher-order correlations for more general cases, such as variable connection weights, heterogeneity in node dynamics, and common input can be treated similarly, using techniques in~\cite{Hu:2012vh}. 

\begin{figure}[H]
\centering
\includegraphics[scale=1.4]{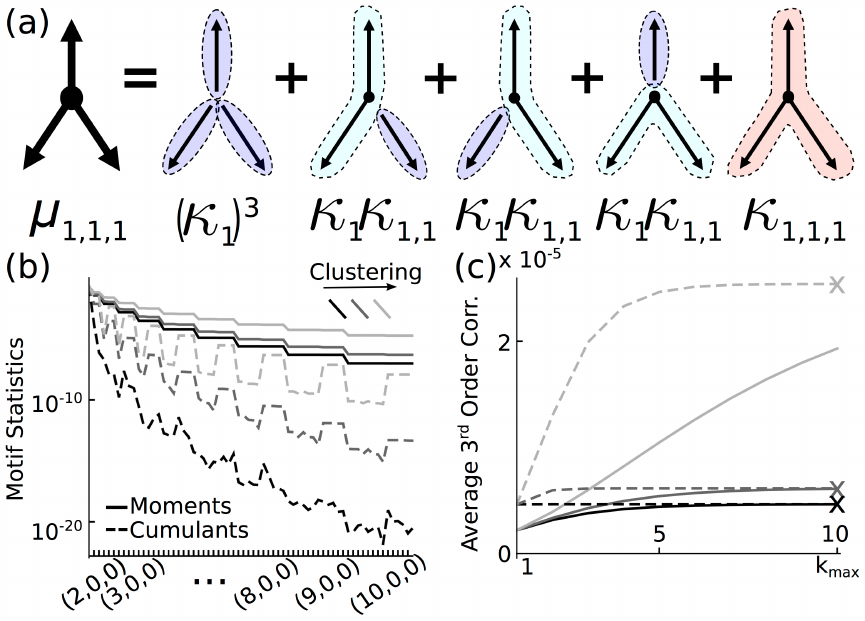}
% switched to .pdf to make compiling faster, change back to .png if needed
\renewcommand{\figurename}{FIG.}
\caption{(a) Cumulant decomposition of a three-branch motif. Panels (b) and (c) are counterparts of Fig.~\ref{f:figure2}(b) and (c) for three-branch motifs and bispectra: (b) decay of motif moments and cumulants with respect to order and (c) convergence estimated third order correlations by two approaches. In (b), the $(n,m,l)$ motifs are again increasingly ordered according to the order $n+m+l$. Because of symmetry, only motifs with $n\ge m\ge l$ are listed. Within each order, motifs are arranged by lexicographical order of $n$, $m$ and $l$, except that motifs with $l\neq 0$ are listed first. }
\label{f:motif_3}
\end{figure}

Thus far, we have shown via Eq.~\eqref{e:cov_epn_hoc} how network motifs --- quantified by the motif moments $\M_{n_1,\cdots,n_k}$ --- contribute to higher-order dynamical correlations.  The solid lines in Fig.~\ref{f:motif_3}(b) show that the motif moments can decay slowly.  The consequence is that motifs of high order (up to 10 or beyond) may be needed for a good approximation of third-order correlations (Fig.~\ref{f:motif_3}(c), solid lines).

It is therefore natural to ask whether the motif cumulant approach can be extended to higher order, and help approximate finer measures of coherence using information about only few lower order motifs. Although  the main ideas are similar as those at second order, derivations  at higher order are more cumbersome. We note that a significant simplification is offered by the  use of our new combinatorial formulation of motif cumulants (Sec.~\ref{S:motif_cumulants}).
%This motivates us to develop a prediction of dynamical coherence that will converge with only motifs of smaller size.  To do this, we next generalize the motif cumulant approach to higher orders dynamical correlation. 

First, we define multi-branch motif cumulants via their relationship with motif moments.
Specifically, we relate the motif moments $\{\M_*\}$ and motif cumulants $\{\K_*\}$ ($*$ stands for multiple indices, see below) via a combinatorial expression. This expression corresponds to the decomposition shown in Fig.~\ref{f:motif_3}(a) (cf. Fig.~\ref{f:figure3}(a)), where we have decomposed a $k$ branch motif into motifs with $k$ and fewer branches.  

We next enumerate all possible ways of decomposing the $k$ branch motif explicitly. Just as in Eqs.~(\ref{E:recursive1},\ref{E:recursive2}), this is done according to how a $k$-branch motif is partitioned at the ``root" of the branches (the sum over $\chi$ in Eq.~\eqref{E:k_cumulant_exp}). In other words, we examine which the $k$-branches are grouped together as one component in the decomposition. To see what this means, examine the coloring in Fig.~\ref{f:motif_3}(a):  for different terms in the decomposition, the components that are shaded with the seam color have been grouped together. The remaining enumeration is about how each branch breaks up (into chains $\K_{B^i_j}$, corresponding to the sum over $\pi_i$ in in Eq.~\eqref{E:k_cumulant_exp}). \beq
\label{E:k_cumulant_exp}
\M_{n_1,\cdots, n_k}= \sum_{\pi_1,\cdots,\pi_k}
\left(\prod_{i=1}^k
 \prod_{j=2}^{t_i}\K_{B^i_j}\right) 
 \left( \sum_{\chi} \prod_{\{i_1,\cdots,i_s\} \in \chi} \K_{B_1^{i_1},\cdots, B_1^{i_s}} \right)
\eeq
Here $\pi_i =\{B^i_1,\cdots,B^i_{t_i}\}$ is an ordered partition of $n_i$, $\chi$ is a partition of the set $\{1,\cdots, k\},$ and $\{i_1,\cdots,i_s\}$ is one subset of indices that are grouped according to the partition $\chi$.

To generalize Eq.~\eqref{e:rsm}  we use the following analog of Prop. 4.1 in~\cite{Hu:2012vh}.

\begin{theorem}
\label{TH:resumming_k}
For a pair of motif moments and cumulants $\{\M_*\}$ and $\{\K_*\}$ with up to $k$ branches,
\beq
\begin{split}
&\sum_{n_1,\cdots,n_k=0}^{\infty} \M_{n_1,\cdots,n_k}  =\left(1-\sum_{n=1}^{\infty} \K_{n} \right)^{-k}\left(\sum_{\pi} 
f(\pi) \prod_{B\in \pi, B>1} 
\left( \sum_{n_1,\cdots, n_B=1}^{\infty} \K_{n_1,\cdots,n_B}  \right)
\right),
\end{split}
\label{e.thmsum}
\eeq
assuming all series converge absolutely and $\vert \sum_{n=1}^{\infty} \K_{n}\vert<1$. The sum with index $\pi$ is through all partitions of $k$. When indices for the product are empty, we take the corresponding terms to be 1. The Fa\`{a} di Bruno coefficient
\beq
f(\pi)=\left(\sum_{B\in \pi} B\right)! \cdot \left(\prod_{B\in \pi}  B! \right)^{-1}
\cdot  \left(\prod_{B\in \pi}  \#B! \right)^{-1}
\eeq 
is the number of partitions of set $\{1,\cdots, k\}$ that correspond to a partition $\pi$ of integer $k$. Here $\nu$ is the set of unique $B$'s in $\pi$, and for every $B\in \nu$, $\#B$ is the number of repetitions of $B$ in $\pi$. 
\end{theorem}

We provide a proof of Theorem \ref{TH:resumming_k} in Appendix~\ref{S:hoc} using the combinatorial relation given by Eq.~\eqref{E:k_cumulant_exp}. We note that the proof itself is different from the matrix based method used in~\cite{Hu:2012vh} to obtain the second order correlation result. Moreover, this new approach can be easily generalized to the case of subpopulations  (see below and Appendix~\ref{S:sub_pop_rsm}).

To establish an expression for average higher order correlations, first note that Eq.~\eqref{E:k_cumulant_exp} is ``homogeneous in degree," so that if it is satisfied for a pair of motif moments and cumulants $\{\M_*\}$, $\{\K_* \}$, it will also be satisfied for scaled pairs $\{g^{\vert *\vert} \M_*\}$, $\{g^{\vert *\vert} \K_* \}$.  Thus, the same relationship holds for scaled motif statistics.  Applying Theorem~\ref{TH:resumming_k} to Eq.~\ref{e:cov_epn_hoc}, using scaled motif statistics, we obtain

\beq
\label{E:rsm_hoc}
\begin{split}
&\frac{\langle\bfS_{y[k]} \rangle}{S_{x[k]}}
=\frac{1}{N^{k-1}}\left(1-\sum_{n=1}^{\infty} g^n \K_{n} \right)^{-k}\left(\sum_{\pi} 
f(\pi)
\prod_{B\in \pi, B>1} 
\left( \sum_{n_1,\cdots, n_B=1}^{\infty} g^{\vert n \vert} \K_{n_1,\cdots,n_B}  \right)
\right).
\end{split}
\eeq

As an example, the motif cumulant expansion of the average third order correlation is
\beqr
\label{e:rsm_3}
&& \frac{\langle\bfS_{y[3]} \rangle}{S_{x[3]}}
= \frac{1}{N^2}
\left(1 - \sum_{n=1}^\infty g^n \K_n  \right)^{-3}
\left(1 + 3 \sum_{l,m=1}^\infty  g^{l+m}  \K_{l,m}
+\sum_{l,m,n=1}^\infty  g^{l+m+n}  \K_{l,m,n} \right).
\eeqr

Fig.~\ref{f:motif_3}(b-c) are counterparts of Fig.~\ref{f:figure2}(b-c) that numerically compare motif moment and cumulant approaches for stochastic block networks. They show numerically that our observations for pairwise correlations generalize to higher orders (see also  Fig.~\ref{f:figure2_3rd} in the Appendix for an application to the \BA{} network). First, we show that higher order correlations can depend  on long
paths through the network (motif moments, solid lines). Second, when predicting average correlation using motif statistics up to a given order, an approximation in terms of motif cumulants is more accurate than one in terms of motif moments (panel (c)).  Third, the order of motif statistics needed to approximate correlations again increases with network homogeneity (compare lines of different shade). 

Finally, the subpopulation approach generalizes to higher order.  The resulting general formula is given in Corollary~\ref{c:rsm_hoc_sub}, and derived in Appendix~\ref{S:multi_HOC}.  Moreover, this result offers similar advantages in predicting correlations from lower-order motif cumulants (see Appendix \ref{S:sub_pop_rsm} and Fig.~\ref{f:figure3_3rd}).

\section{Conclusion}

Network motifs have been used previously to  link  local network connectivity and global coherence in networks with linearly interacting components~\cite{Pernice:2011,Trousdale:2012,Hu:2012vh}.  Here, we 
developed this theory in order to make it both more general and more broadly applicable.  We first showed that a motif-based approach
 introduced in prior work has a  probabilistic interpretation in terms of quantities closely related to key statistical concepts.  We refer to these as motif cumulants. 
 
 Next, we showed that the link between network architecture and dynamical correlation -- through motif cumulants -- can be {\it complex} in clustered and heterogeneous networks.  This complexity can result in the apparently irreducible contribution of long paths  to network-wide coherence.  However, the motif cumulant approach can be extended to reduce this complexity -- and hence the size and number of the network features that must be sampled empirically -- substantially. Finally, we showed how the theory naturally extends to higher-order dynamical correlations, for a broad subset of the dynamical models under study.  This provides a direct link between local network architecture and global dynamics at every order.  

An important feature of our approach for experimental settings is that the prevalence of only a limited number of motifs is needed in order to predict network-wide dynamical coherence.  Moreover, these motifs are small, involving only a few nodes at a time.  This property could provide a way forward in experimental settings --  as in studies of networks of genes \cite{Alon:2007uu} or neurons \cite{Song:2005, Perin:2011} -- in which networks are quantified by sampling a limited number of edges measured simultaneously.  The resulting motif prevalences are precisely the quantities needed to define the motif cumulants that are at the core of our approach.

The present results suggest many opportunities for future research.  At the top of the list is extending the connection between network motifs and higher-order dynamical correlations to apply to coupled point process models.  Somewhat surprisingly, we have found both numerically and  analytically (in special cases) that the linear response approach (Eq.~\eqref{e:hoc}) that extends to all orders for finite-valued stochastic processes fails to extend beyond second order for coupled point process models, where each node generates ``spike" events (data not shown).  Future research will explore modifications of the linear response approach that may re-establish a useful description of higher order correlations for these network models.  This would open the door to studies of plasticity and learning of network connections in neural systems, where interactions are governed by spike times~\cite{Pfister:2006us}. 

We close by mentioning two further extensions of special interest.  The first concerns applications to stimulus-encoding networks.   Such networks can be heterogeneous and composed of groups of nodes, each with different connectivity rules and, importantly, responding differently to an external stimulus.  Networks with spatial structure provide a natural way in which such connectivity and responses might develop.  For such a network, our subpopulation motif approach could predict the levels of dynamical coherence within and between each group of nodes.  From here, decoding techniques could quantify the level of information that the neural groups carry about the stimulus itself, and how this depends on the correlation structure induced by different network motifs~\cite{Aver:2006,gawne93,Cohen:2011eh,Zohary:1994ei,Sompolinsky:2001hh,Abbott:1999ul,Zylberberg:2012ty,HuZSB:2013}.  

A final open problem concerns the invertibility of the architecture-to-dynamics question considered here.  Given measurements of network-wide coherence, what can we conclude about network architecture?  The network motif approach can narrow  the possibilities, especially when higher-order correlations are considered, but we do not yet know what additional assumptions are required to yield a unique solution.

\section*{Acknowledgements}

\noindent We thank C. Hoffman, K. Bassler and B. Doiron for helpful insights and suggestions.  This work was supported by NSF grant DMS-1122094, a joint NSF/NIGMS grant R01GM104974, and a Texas NHARP award to KJ, and by a Career Award at the Scientific Interface from the Burroughs Wellcome Fund and NSF Grants DMS-1056125 and DMS-0818153 to ESB.

\section*{Appendix}

%htbp
\begin{table}[H]\caption{Notations}
\begin{center}% used the environment to augment the vertical space
% between the caption and the table
\begin{tabular*}{\textwidth}{ l @{\extracolsep{\fill}}  r}
\hline
$N$, $N_\alpha$ & size of the whole population or subpopulation $\alpha$\\
$y_i$, $\bfy$ & activity of node $i$ or that for all nodes combined as a column vector \\
$x_i$, $\bfx$ & baseline activity of node $i$ in the absence of coupling between nodes \\
$A_i(t)$ & linear response kernel of node $i$\\
$\bfW$, $\bfW^0$ & connection and adjacency matrix\\
$w$ & connection strength\\
$g=N\At w$ & effective coupling strength\\
$\bfC_y(\tau)$  & the matrix with cross-covariances of all node pairs of $\bfy$~\cite{corr_def}\\
$\bfC_{y[k]}$ & the $k$-tensor with all $k$-th order correlations of $\bfy$, Eq.~\eqref{E:C_tensor_def}\\
$\bfS_y(\om)$, $\bfS_{y[k]}$ & Frequency domain counterparts of $\bfC_y(\tau)$ and $\bfC_{y[k]}$, \\ & see~\cite{ft_sp_def} and Eq.~\eqref{E:WK_hoc}\\
$\M_{*}$, $\K_{*}$ & motif moment and cumulant, $*$ stands for any subscript describing\\ & the length of branches, such as $n$ or $n,m$, Eq.~(\ref{e:cov_epn_avg}, \ref{E:recursive1}-\ref{E:recursive2},\ref{e:cov_epn_hoc},\ref{E:k_cumulant_exp})\\
$\pi$, $B\in \pi$ & partition (or ordered partition) of an integer and its components\\
$\chi$, $\{i_1,\cdots,i_s\}\in \chi$ & partition of a set $S$ and its components (a subset of $S$)\\ 
$\bfPt=(\bfI-\At\bfW)^{-1}$ & propagation factor\\
$\langle \cdot \rangle$ & empirical average of matrix or \\ & tensor (sum  of entries divided by their number)\\
$\mathcal{F}(z)$ or $\tilde{z}$ & Fourier transform of $z$ (transform taken \\ & entry-wise if $z$ is a matrix or tensor)\\
$\bar {z}$ & complex conjugate\\
$\bfW^T$ & matrix transpose\\
$\Psi(\bfW)$ & spectral radius\\
\hline
\end{tabular*}
\end{center}
\label{T:notation}
\end{table}

\subsection{Relating the  \OU{} model to Eq.~\eqref{e:lin} of the main text}
\label{S:OU_relation}

We used a simplified form of the canonical \OU{} (OU) model in all examples where
we consider second-order statistical quantities. This model is related to Eq.~\eqref{e:lin}
in the main text by writing the dynamics
\beq
\label{e:OU}
\dot{\bfy}=-\bfLambda \bfy(t)+\bfW \bfy(t) +\bfxi(t).
\eeq
where $\bfy(t)=(y_1(t),\cdots,y_N(t))^T$.  The diagonal matrix $\bfLambda=\tau^{-1}\bfI$ sets the intrinsic timescale of the nodes, and 
the column vector $\bfxi(t)$ is composed of independent white noise processes. Eq.~\eqref{e:OU} above is then equivalent to  Eq.~\eqref{e:lin} of the main text with
$A_i(t) = A(t)=e^{-t/\tau}\Theta(t)$.
Upon coupling, the baseline activity of a node in the network, $\bfx_i(t)=(A * \bfxi)(t)$,  
is perturbed by filtered input from other nodes, $A*\sum_j\bfW_{ij} y_j(t)$.

\subsection{Further examples}
\label{S:more_eg}

Here we provide details of several computational findings referred to in the main text.    Each addresses the generality and applicability of our results.  First, Fig.~\ref{f:figure2_SF} shows that our main results contrasting motif moments and cumulants hold for the \BA{} network model, which has significantly more complex structure than the stochastic block models studied in Fig.~\ref{f:figure2} of the main text.

Next, Figs.~\ref{f:figure2_3rd} and~\ref{f:figure3_3rd} present analogous results for third-order correlations in network output.  Specifically, Fig.~\ref{f:figure2_3rd} shows that these third-order correlations depend significantly on the details of the underlying graph structure (i.e., the degree of clustering).  Moreover, this dependence can be efficiently predicted via motif cumulants. Fig.~\ref{f:figure3_3rd} demonstrates that the subpopulation approaches continue to enhance the accuracy of our predictions -- if the populations are correctly defined, levels of triplet correlations can be predicted from lower-order motifs.

Fig.~\ref{F:PF_vec} provide numerical evidence for our claim that the PF vector for a general class of networks is closely approximated by the degree list (see Sec.~\ref{S:heter_sub}.\ref{S:partition}).

\begin{figure}[H]
\centering
\includegraphics[scale=1.4]{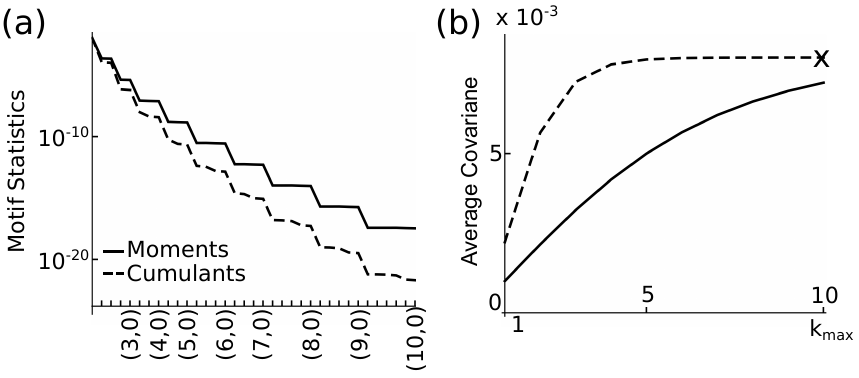}
\renewcommand{\figurename}{FIG.}
\caption{Same as Fig.~\ref{f:figure2}(b,c) of the main text but for the \BA{} model. (a) The magnitude of motif cumulants (dashed lines) and moments (solid lines) for a \BA{} model network. Motifs $(n,m),~n\ge m$ are grouped first by $n+m$ and then  arranged by increasing $n$. (b) Approximations of average covariances using motif moments (truncating Eq.~\eqref{e:cov_epn_avg}, solid lines) and cumulants (truncating Eq.~\eqref{e:rsm}, dashed lines) up to order $k_{\max}$. Exact values (direct evaluation of Eq.~\eqref{e:cov}) are labeled by crosses: a \BA{} network of size 1000 and connection probability 0.01.}
\label{f:figure2_SF}
\end{figure}

\begin{figure}[H]
\centering
\includegraphics[scale=1.4]{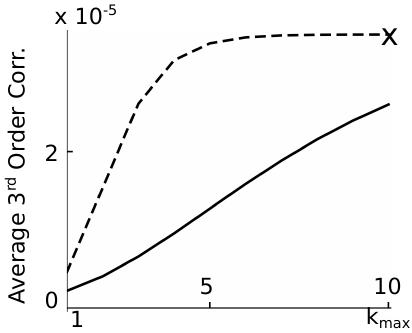}
\renewcommand{\figurename}{FIG.}
\caption{Same plot as Fig.~\ref{f:figure2_SF}(b) but for average third order correlations $\langle \bfS_{y[3]}\rangle / S_{x[3]}$.  Approximations using motif moments (solid lines) and cumulants (dashed lines) up to order $k_{\max}$ for a \BA{} network of size 1000 and connection probability 0.01. }
\label{f:figure2_3rd}
\end{figure}

\begin{figure}[H]
\centering
\includegraphics[scale=1.4]{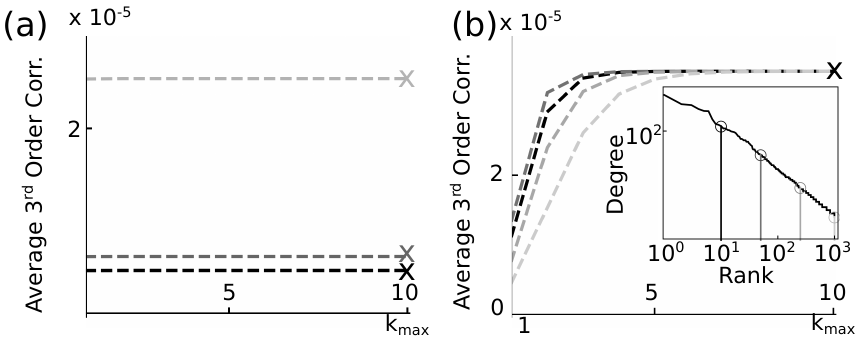}
\renewcommand{\figurename}{FIG.}
\caption{Same plots as Fig.~\ref{f:figure3} of the main text but for average third order correlations $\langle \bfS_{y[3]}\rangle / S_{x[3]}$. Approximations using the subpopulation cumulant approach by truncating at order $k_{\max}$, the exact values (direct evaluation of Eq.~\eqref{e:hoc} of the main text) are labeled by crosses: (a) stochastic block model networks of Fig.~\ref{f:figure2}  of the main text (the colors are the same) divided into two subpopulations with differing connectivities; (b) the \BA{} network of SI Fig.~\ref{f:figure2_SF} divided into two subpopulations according to different thresholds on the sum of in- and out-degrees  (different colors, see also the inset, which displays the cutoffs).}
\label{f:figure3_3rd}
\end{figure}

\begin{figure}[H]
\centering
\includegraphics[width=4.7in]{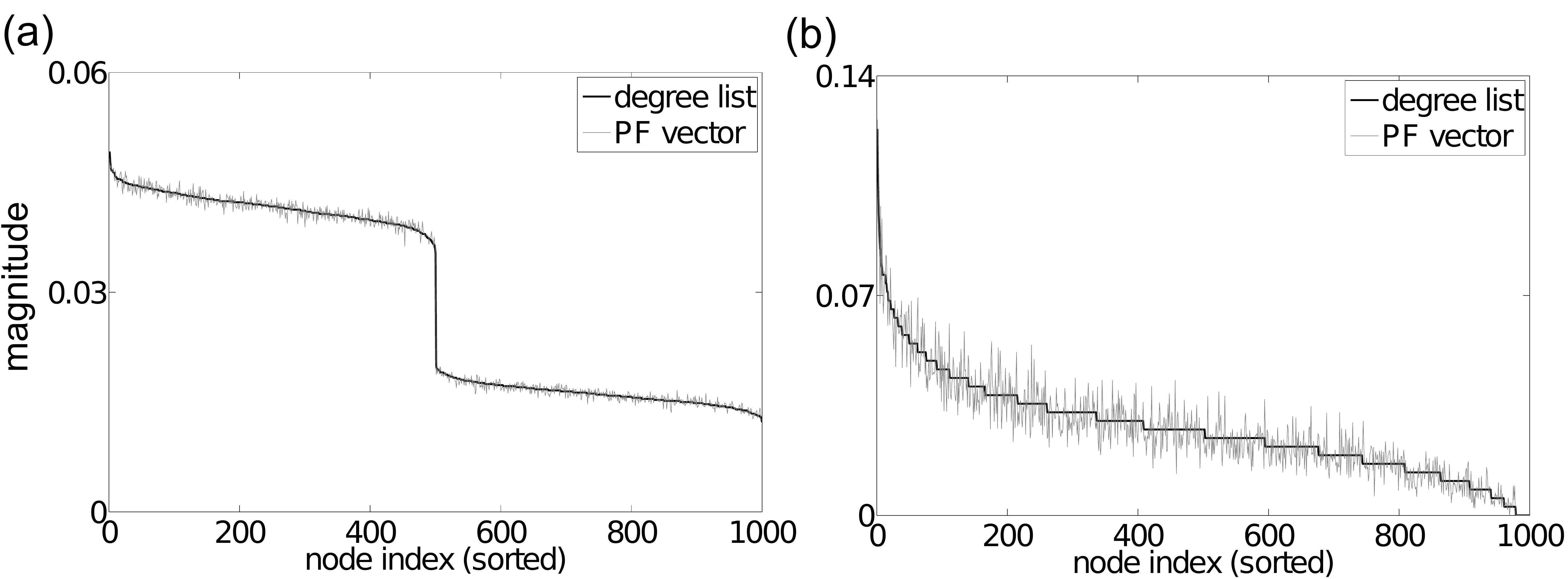}
\renewcommand{\figurename}{FIG.}
\caption{Comparing the PF vector (gray) and the in-degree list (black) for a stochastic block network (a) and a \BA{} network (b). The vectors are normalized to have unit $L_2$ norm and plotted by the indices of nodes, which are ordered in descending in-degree. The stochastic block network has $s1=1.44\sqrt{0.2}$ and is the same one in Fig.~\ref{f:figure2} (b-c). The \BA{} network is the same one in Fig.~\ref{f:figure3} (b).}
\label{F:PF_vec}
\end{figure}

\subsection{Details of numerical results}
\label{S:num_details}

Here we provide a detailed description of the computational examples in the main text and the appendix. This includes all parameters describing the dynamics of nodes and connections, and our methods of generating random networks.

In Fig.~\ref{F:1} of the main text, we calculated correlations for an OU system (see Eq.~\eqref{e:OU} of the main text) with $\tau = 1$, $\bfxi$ having unit intensity, and
$$
\bfW = \left[\begin{matrix}0 & -0.75 \\ -0.75 & 0 \end{matrix}\right].
$$

In plots of approximations of average second and third order covariances, i.e. Fig.~\ref{f:figure2}(b), \ref{f:figure3}(a,b)  of the main text, SI Figs.~\ref{f:figure2_SF}(b), \ref{f:figure2_3rd}(a,b), and \ref{f:figure3_3rd}(a,b), the parameters $\At$ and $w$ are chosen so that $N\At w p=0.4$. Note that the choice of $S_x$ (resp. $S_{x[3]}$ at third order) will not affect the normalized quantity $\langle \bfS_y \rangle /S_x$ (resp. $\langle \bfS_{y[3]} \rangle /S_{x[3]}$), and can be set to 1.

The \BA{} networks in Fig.~\ref{f:figure3}(b), and SI Figs.~\ref{f:figure2_SF}(a,b), \ref{f:figure2_3rd}(b), and \ref{f:figure3_3rd}(b) are generated by a directed \BA{} model similar to that in \cite{Prettejohn:2011tu}. One starts with a ``core" of $Np$ nodes, randomly connected with connection probability 0.5. After that, $N-Np$ nodes are added to the graph. When adding a new node $i+1$, it will form exactly $Np$ connections with the existing nodes $1,\cdots, i$. Those connections are distributed among existing nodes according to probabilities that are proportional to the sum of in- and out- degree of each node. The direction of the connection, whether into node $i+1$ or out of node $i+1$, is chosen independently with probability 0.5. The code implementing this algorithm is available upon request.

\subsection{Explicit expressions for motif cumulants}
\label{S:exp_motif}

Here, we will prove that the following matrix expressions for $\K_n$ and $\K_{n,m}$ introduced in~\cite{Hu:2012vh} are equivalent to the recursive definition in Eqs.~(\ref{E:recursive1},\ref{E:recursive2}) of the main text:
\begin{equation}\label{E:explicit1}
\begin{split}
\K_{n} &= \frac{1}{N^{n+1}}\sum_{i,j}(\underbrace{\bfW^0 \bfTheta \bfW^0 \cdots \bfTheta \bfW^0}_{n \mathrm{\ factors\ of\ }   \bfW^0})_{ij}\\
&= \frac{1}{N^{n}} \bfu^T \left[ \left(\bfW\bfTheta\right)^{n-1} \bfW\right] \bfu\\
&= \frac{1}{N^{n}}\bfu^T \bfW^\theta_{n}\bfu,
\end{split}
\end{equation}
\begin{equation}\label{E:explicit2}
\begin{split}
\K_{n,m} &= \frac{1}{N^{n+m+1}}\sum_{i,j}(\underbrace{\bfW^0 \bfTheta \bfW^0 \cdots \bfTheta \bfW^0}_{n \mathrm{\ factors\ of\ }   \bfW^0} \bfTheta \underbrace{\bfW^{0T} \bfTheta \bfW^{0T} \cdots \bfTheta \bfW^{0T}}_{m \mathrm{\ factors\ of\ } \bfW^{0T}})_{ij}\\
&= \frac{1}{N^{n+m}} \bfu^T \left[ \left(\bfW\bfTheta\right)^{n-1} \bfW\bfTheta\bfW^T \left(\bfTheta\bfW^T\right)^{m-1}\right] \bfu\\
&= \frac{1}{N^{n+m}}\bfu^T \bfW^\theta_{n} \bfTheta \bfW_m^\theta\bfu,
\end{split}
\end{equation}
where
\[
 \bfW^\theta_{n} = \left[\bfW^0 \bfTheta\right]^{n-1}\bfW^0
\]
and $\bfu=(1,\cdots,1)^T/\sqrt{N}$, $\bfH=\bfu \bfu^T$, $\bfTheta=\bfI-\bfH$.

We see that $\bfW^0 \bfTheta$, $\bfTheta \bfW^T$ are recurring factors in $\K_n$ and $\K_{n,m}$. Using the relation of spectral radius and matrix norm~\cite{Pernice:2011}, one can show that the asymptotic decay speed of $\K_{*}$ is determined by the spectral radii of these factors. Interestingly, it is easy to show that $\Psi(\bfW^0 \bfTheta)=\Psi(\bfTheta \bfW^0\bfTheta)=\Psi(\Theta \bfW^0)$ hence these spectral radii coincide. A similar argument relates the decay of $\M_{*}$ with $\Psi(\bfW^0)$ (Eq.~\eqref{E:mn_explicit}).

We prove only that Eq.~\eqref{E:explicit2} holds, since a nearly identical, but simpler, proof verifies Eq.~\eqref{E:explicit1}. First, recalling that $\M_{n,m} = \langle \bfW^{0n}\left(\bfW^{0m}\right)^T \rangle / N^{n+m-1}$, it is straightforward to show that
\begin{equation}\label{E:mn_explicit}
\M_{n,m} = \frac{1}{N^{n+m}}\bfu^T (\bfW^0)^n (\bfW^{0T})^m \bfu.
\end{equation}
Substituting $\bfI = \bfTheta+\bfH$ between every subsequent appearance of the adjacency matrix $\bfW^0$ gives
\begin{equation}\label{E:mn_explicit_sub}
\M_{n,m} = \frac{1}{N^{n+m}}\bfu^T \left[\bfW^0(\bfTheta+\bfH)\right]^{n-1}\bfW^0(\bfTheta+\bfH)\bfW^{0T} \left[\left(\bfTheta+\bfH\right)\bfW^{0T}\right]^{m-1} \bfu.
\end{equation}
By expanding across all sums of $\bfTheta+\bfH$ except the central one (between the terms $\bfW^0,\bfW^{0T}$), and noting that there is an obvious bijection between a pair of compositions of the integers $n$ and $m$, \emph{i.e.}, $\{n_1,\ldots,n_t\} \in \mathcal{C}(n), \{m_1,\ldots,m_s\} \in \mathcal{C}(m)$, and a term of the form
$$
\left[ \prod_{i=1}^{t-1} \left( \bfW^\theta_{n_i} \bfH \right) \right] \left[  \bfW^\theta_{n_t}(\bfTheta+\bfH) \bfW^\theta_{m_s}\right]\left[ \prod_{j=1}^{s-1} \left( \bfH \bfW^\theta_{m_j}  \right) \right]
$$
we may write (using $\bfH=\bfu\bfu^T$)
\begin{equation}\label{E:cum_mom_relation}
\begin{split}
\M_{n,m} &= \frac{1}{N^{n+m}}\bfu^T \left\{ \sum_{\substack{\{ n_1,\ldots,n_t\} \in \mathcal{C}(n) \\ \{ m_1,\ldots,m_s\} \in \mathcal{C}(m)}}\left[ \prod_{i=1}^{t-1} \left( \bfW^\theta_{n_i} \bfH \right) \right] \left[  \bfW^\theta_{n_t}(\bfTheta+\bfH) \bfW^\theta_{m_s}\bfu\right]\left[ \prod_{j=1}^{s-1} \left( \bfH \bfW^\theta_{m_j}  \right) \right]\right\}\bfu\\
&=\frac{1}{N^{n+m}} \sum_{\substack{\{ n_1,\ldots,n_t\} \in \mathcal{C}(n) \\ \{ m_1,\ldots,m_s\} \in \mathcal{C}(m)}}\left[ \prod_{i=1}^{t-1} \left( \bfu^T \bfW^\theta_{n_i} \bfu \right) \right] \left[ \bfu^T \bfW^\theta_{n_t}(\bfTheta+\bfu\bfu^T) \bfW^\theta_{m_s}\bfu\right]\left[ \prod_{j=1}^{s-1} \left( \bfu^T \bfW^\theta_{m_j} \bfu \right) \right]\\
&= \sum_{\substack{\{ n_1,\ldots,n_t\} \in \mathcal{C}(n) \\ \{ m_1,\ldots,m_s\} \in \mathcal{C}(m)}}\left[ \prod_{i=1}^{t-1} \left(  \frac{1}{N^{n_i}}\bfu^T \bfW^\theta_{n_i} \bfu \right) \right] \left[\frac{1}{N^{n_t+m_s}} \bfu^T \bfW^\theta_{n_t}(\bfTheta+\bfu\bfu^T) \bfW^\theta_{m_s}\bfu\right] \\
&\hspace{2in}\cdot\left[ \prod_{j=1}^{s-1} \left( \frac{1}{N^{m_j}}\bfu^T \bfW^\theta_{m_j} \bfu \right) \right]
\end{split}
\end{equation}
If $t=1$, we define the product $\left[ \prod_{i=1}^{t-1} \left( \bfW^\theta_{n_i} \bfH \right) \right] = \bfI$.

\par
We now prove Eq.~\eqref{E:explicit2} by induction, assuming Eq.~\eqref{E:explicit1} holds. First, when $n = m = 1$, the only compositions are trivial (i.e., $\pi_1 = \pi_2 = \{1\}$). Equating in this case the right-hand sides of Eq.~\eqref{E:recursive2} of the main text and  Eq.~\eqref{E:cum_mom_relation} gives that
$$
 \K_{1,1} + \left(\K_1\right)^2 = \frac{1}{N^2}\bfu^T \bfW_1^\theta \bfTheta \bfW_1^\theta \bfu + \left(\frac{1}{N}\bfu^T\bfW_1^\theta\bfu\right)^2.
$$
Since Eq.~\eqref{E:explicit1} for $n=1$ gives that
$$
\K_1 = \frac{1}{N}\bfu^T\bfW_1^\theta\bfu,
$$
we have that Eq.~\eqref{E:explicit2} holds for $n=m=1$. Next, assume Eq.~\eqref{E:explicit2} is true for all $ (p,q)$ such that $p\leq n$ and $q<m$ \emph{or} $p < n$ and $q \leq m$. That is, in these cases,
$$
\K_{p} = \frac{1}{N^{p}}\bfu^T \bfW^\theta_{p} \bfu \ \text{(by Eq.~\eqref{E:explicit1})}\quad\text{and}\quad \K_{p,q} = \frac{1}{N^{p+q}} \bfu^T \bfW^\theta_{p}\bfTheta\bfW^\theta_{q}\bfu.
$$

Making the corresponding substitutions in Eq.~\eqref{E:cum_mom_relation}, the only term we have not accounted for in matching the right-hand side of Eq.~\eqref{E:cum_mom_relation} to that of Eq.~\eqref{E:recursive2} of the main text are the terms corresponding to the pair of compositions $\{n\},\{m\}$. In Eq.~\eqref{E:recursive2} of the main text, the corresponding terms are
\begin{equation}\label{E:cmr_term1}
\K_{n,m} + \K_n \K_m
\end{equation}
while in Eq.~\eqref{E:cum_mom_relation}, the terms take the form
\begin{equation}\label{E:cmr_term2}
\begin{split}
\frac{1}{N^{n+m}}  \bfu^T \bfW^\theta_{n}(\bfTheta+\bfu\bfu^T) \bfW^\theta_{m}\bfu &= \frac{1}{N^{n+m}} \bfu^T \bfW^\theta_n \bfTheta \bfW^\theta_m \bfu +  \left(\frac{1}{N^{n}} \bfu^T \bfW^\theta_n \bfu\right)\left(\frac{1}{N^m}\bfu^T \bfW^\theta_m \bfu\right)\\
&= \frac{1}{N^{n+m}} \bfu^T \bfW^\theta_n \bfTheta \bfW^\theta_m \bfu + \K_n \K_m,
\end{split}
\end{equation}
where the second equality follows from the inductive assumption. Comparing Eqs.~(\ref{E:cmr_term1},\ref{E:cmr_term2}) gives that 
$$
\K_{n,m} =  \frac{1}{N^{n+m}} \bfu^T \bfW^\theta_n \bfTheta \bfW^\theta_m \bfu ,
$$
which is exactly Eq.~\eqref{E:explicit2}, completing the inductive proof.

%\subsection{Expression of $f(\pi)$}
%\label{S:exp_f_coef}
%
%Let $\nu$ be the set of unique $B$'s in $\pi$, and for every $B\in \nu$, $\#B$ is the number of repetitions of $B$ in $\pi$. 
%\beq
%f(\pi)=\left(\sum_{B\in \pi} B\right)! \cdot \left(\prod_{B\in \pi}  B! \right)^{-1}
%\cdot  \left(\prod_{B\in \pi}  \#B! \right)^{-1}.
%\eeq

\subsection{Proof of  the theorem on multi-branch motifs}
\label{S:proof_hoc_rsm}

\begin{proof}[Proof of Theorem~\ref{TH:resumming_k}]

First, we rewrite the LHS of Eq.~\eqref{e.thmsum}, to explicitly account for cases with different $n_j$ nonzero.  Specifically, we sum over all possible sets of $k \le k'$ indices $\{j_1,\cdots,j_k\}$ corresponding to the nonzero values of $n_j$:
\beq
\sum_{n_1,\cdots,n_{k^\prime}=0}^{\infty} \M_{n_1,\cdots,n_{k^\prime}} 
=
\sum_{\substack{ \{j_1,\cdots,j_k\} \\\subset \{1,\cdots, k^\prime\}}} \sum_{n_{j_1},\cdots,n_{j_k}=1}^{\infty} \M_{n_{j_1},\cdots,n_{j_k}}
\label{e.nonzero}
\eeq

We now focus on a fixed $k$, and without loss of generality, let $\{j_1,\cdots,j_k\} =\{1,\cdots,k\}$.  Applying Eq.~\eqref{E:k_cumulant_exp}, we have
\beqrn
\sum_{n_1,\cdots,n_k=1}^{\infty} \M_{n_1,\cdots,n_k}&=&
\sum_{n_1,\cdots,n_k=1}^{\infty}\sum_{\pi_1,\cdots,\pi_k}\left(\prod_{i=1}^k
 \prod_{j=2}^{t_i}\K_{B^i_j}\right) 
 \left( \sum_{\chi} \prod_{\{i_1,\cdots,i_s\} \in \chi} \K_{B_1^{i_1},\cdots, B_1^{i_s}} \right)\\
 &=&
 \sum_{\chi} \sum_{B_1^{i_1},\cdots ,B_1^{i_s} \ge 1}^{\infty} \prod_{\{i_1,\cdots,i_s\} \in \chi} \K_{B_1^{i_1},\cdots, B_1^{i_s}}
 \left(  \sum_{n_1,\cdots,n_k}^{\infty}\sum_{\pi_1,\cdots,\pi_k}
\left(\prod_{i=1}^k
 \prod_{j=2}^{t_i}\K_{B^i_j}\right)   \right).
\eeqrn
In the last equality, we switched the order of summations by pulling the sum over $\chi$ to the front. Consequently, for each fixed $\{B_1^{i_1},\cdots ,B_1^{i_s}\}$ taken in an outer sum, the $\{\pi_i\}$ and $\{n_i\}$ are restricted to terms that are possible for that $\{B_1^{i_1},\cdots ,B_1^{i_s}\}$. Notice that these sums can be factorized as
\beq
\sum_{n_1,\cdots,n_k}^{\infty}\sum_{\pi_1,\cdots,\pi_k}
\left(\prod_{i=1}^k
 \prod_{j=2}^{t_i}\K_{B^i_j}\right)  
 =\prod_{i=1}^k \left( \sum_{n_i \ge B^{i}_1}^\infty  \sum_{\pi_i \ni B_1^{i}} \prod_{j=2}^{t_i} \K_{B^i_j} \right).
\label{e.reordered}
\eeq

We next will simplify the factors on the RHS of Eq.~\eqref{e.reordered}.  First, we shift the (dummy) indices of summation and multiplication to explicitly begin counting at the second block in the branch:  
\beqrn
&&\sum_{n_i \ge B^{i}_1}^\infty  \sum_{\pi_i \ni B_1^{i}} \prod_{j=2}^{t_i} \K_{B^i_j}
=\sum_{n_i^\prime=0}^\infty  \sum_{\pi_i^\prime } \prod_{j=1}^{t_i^\prime} \K_{B^i_{j+1}},
\eeqrn
where $n_i^\prime=n_i-B^i_1$, $\pi_i^\prime=\pi_i\backslash \{B_1^i\}$ and $t_i^\prime=t_i-1$ as we exclude the $B_{1}^i$ component. For simplicity, we will drop the primes in the summation indices, and then let $B^{i}_j$ range over the components of the resulting partition (thus rewriting $j+1 \rightarrow j$ below).  Doing this, and further rearranging the terms, we have
\beqrn
&&\sum_{n_i=0}^\infty  \sum_{\pi_i } \prod_{j=1}^{t_i} \K_{B^i_j}
=\sum_{t_i=0}^\infty  \sum_{n_i\ge t_i}\sum_{\pi_i}  \prod_{j=1}^{t_i}\K_{B^i_j}\\
&=&\sum_{t_i=0}^\infty  \prod_{j=1}^{t_i} \sum_{B^i_j=1}^\infty \K_{B^i_j}
=\sum_{t_i=0}^\infty  \left( \sum_{n_i=1}^\infty \K_{n_i} \right)^{t_i}
=\left(1-\sum_{n=1}^{\infty} \K_{n} \right)^{-1},
\eeqrn
where we have summed the geometric series in the last inequality (note the convergence criterion in the Theorem statement).

Therefore, Eq.~\eqref{e.reordered} and the expression above it yield
\beqrn
\sum_{n_1,\cdots,n_k=1}^{\infty} \M_{n_1,\cdots,n_k} 
&=&
 \sum_{\chi} \sum_{B_1^{i_1},\cdots ,B_1^{i_s} \ge 1}^{\infty} \prod_{\{i_1,\cdots,i_s\} \in \chi} \K_{B_1^{i_1},\cdots, B_1^{i_s}}
\left(1-\sum_{n=1}^{\infty} \K_{n} \right)^{-k}\\
&=&  \left(1-\sum_{n=1}^{\infty} \K_{n} \right)^{-k} 
\sum_{\chi} 
\prod_{\{i_1,\cdots,i_s\} \in \chi} 
\left(   \sum_{B_1^{i_1},\cdots ,B_1^{i_s} \ge 1}^{\infty}   \K_{B_1^{i_1},\cdots, B_1^{i_s}} \right) \\
&=&  \left(1-q \right)^{-k} 
\sum_{\chi} 
\prod_{\{i_1,\cdots,i_s\} \in \chi} 
\left(   \sum_{B_1^{i_1},\cdots ,B_1^{i_s} \ge 1}^{\infty}   \K_{B_1^{i_1},\cdots, B_1^{i_s}} \right), 
\eeqrn
where we let $q=\sum_{n=1}^\infty \K_n$.  The above gives a useful expression for the sum over all motifs with exactly $k$ branches of nonzero length.  To establish the theorem, we use this expression for different subsets of $\{1,\cdots, k^\prime\}$ (hence different $k$) that occur in Eq.~\eqref{e.nonzero}. Doing this, we have
\beqrn
&&\sum_{n_1,\cdots,n_{k^\prime}=0}^{\infty} \M_{n_1,\cdots,n_{k^\prime}} 
=
\sum_{\substack{ \{j_1,\cdots,j_k\} \\\subset \{1,\cdots, k^\prime\}}}  
\sum_{\chi_k} 
\left(1-q\right)^{-k}
\prod_{\{i_1,\cdots,i_s\} \in \chi_k} 
\left(\sum_{B_1^{i_1},\cdots ,B_1^{i_s} \ge 1}^{\infty}   \K_{B_1^{i_1},\cdots, B_1^{i_s}} \right) \\
\eeqrn
where $\chi_k$ is a partition of the set $\{j_1,\cdots,j_k\}$ (through we only use the subscript $k$, $\chi_k$ should actually depend on the set $\{j_1,\cdots,j_k\}$). We next rearrange this expression.  First, we define a lift of each partition $\chi_k$ to a partition $\chi$ of the set $\{1,\cdots,k^\prime\}$, by adding any indices not present in $\chi_k$ as individual groups $\{i_{k+1}\}, \cdots, \{i_{k^\prime}\}$. Next we split the sum across $\{j_1,\cdots,j_k\}$ and $\chi_k$ according to their resulting lift $\chi$.  This creates an outer sum; here, the range of $\chi$ is all possible partitions of $\{1,\cdots,k^\prime\}$.  Thus, the expression above
\beqrn
&=&
\sum_{\chi} 
\sum_{\chi_k \vert \chi}
(1-q)^{-k}
\prod_{\substack{\{i_1,\cdots,i_s\} \in \chi_k}} 
\left(  \sum_{B_1^{i_1},\cdots ,B_1^{i_s} \ge 1}^{\infty}   \K_{B_1^{i_1},\cdots, B_1^{i_s}} \right)
\eeqrn
The inner sum is over all $\chi_k$, $0\leq k \leq k^\prime$ whose lift is $\chi$. We can pull out all factors associated with groups in $\chi_k$ that has only 1 element. Note each of such group $\{i_r\}$ corresponds to a factor $\sum^\infty_{B^{i_r}_1\ge 1} \K_{B^{i_r}_1}=q$. Therefore the rest factors in $\prod_{\{i_1,\cdots,i_s\} \in \chi_k}$ are $(k-m_2)$ $q$-factors, where $m_2$ is the number of indices that are partitioned into a group with more than 1 elements in $\chi_k$ (or $\chi$).  Thus, the expression above
\beqrn
&=&
\sum_{\chi} 
\sum_{\chi_k \vert \chi}
(1-q)^{-k} q^{k-m_2}
\prod_{\substack{\{i_1,\cdots,i_s\} \in \chi_k,\\ s>1}} 
\left(  \sum_{B_1^{i_1},\cdots ,B_1^{i_s} \ge 1}^{\infty}   \K_{B_1^{i_1},\cdots, B_1^{i_s}} \right)\\
&=&
\left(1-q \right)^{-k^\prime} 
\sum_{\chi} 
\prod_{\substack{\{i_1,\cdots,i_s\} \in \chi,\\ s>1}} 
\left(  \sum_{B_1^{i_1},\cdots ,B_1^{i_s} \ge 1}^{\infty}   \K_{B_1^{i_1},\cdots, B_1^{i_s}} \right)
\cdot
\left(
 \sum_{\chi_k \vert \chi} 
q^{k-m_2}  (1-q)^{k^\prime-k}
\right).
\eeqrn

For a fixed $k$, it is easy to see the number of $\chi_k$ whose lift being $\chi$ is $\binom{k^\prime-m_2}{k-m_2}$. Hence
\beq
\nonumber
 \sum_{\chi_k \vert \chi} 
q^{k-m_2}  (1-q)^{k^\prime-k}
=\sum_{k=0}^{k^\prime} 
 \binom{k^\prime-m_2}{k-m_2}
q^{k-m_2}  (1-q)^{k^\prime-k}=1,
\eeq
Finally, the expression above
\beqrn
&=&
\left(1-q \right)^{-k^\prime} 
\sum_{\chi} 
\prod_{\substack{\{i_1,\cdots,i_s\} \in \chi,\\ s>1}} 
\left(  \sum_{B_1^{i_1},\cdots ,B_1^{i_s} \ge 1}^{\infty}   \K_{B_1^{i_1},\cdots, B_1^{i_s}} \right)\\
&=&
\left(1-q \right)^{-k^\prime} 
\sum_{\pi}  f(\pi)
\prod_{B \in \pi, B>1} 
\left( \sum_{n_1,\cdots, n_B=1}^{\infty} \K_{n_1,\cdots,n_B}  \right).
\eeqrn

In the last line, since the factor with $\K_*$ is the same as long as $s$ is the same, regardless of actual value of $i_1,\cdots,i_s$, we switched from summing over set partitions $\chi$ to corresponding integer partitions $\pi$ of $k^\prime$.  This introduces the $f(\pi)$ factor and finishes the proof (see Eq.~\eqref{e.thmsum}).
\end{proof}

\subsection{Establishing the subpopulation cumulant Corollary  \eqref{E:rsm_sub}}
\label{S:sub_pop_rsm}

Beyond the similarity in appearance between the single population and subpopulation formulas Eq.~\eqref{e:rsm} and \eqref{E:rsm_sub}, these two can be precisely connected. Define a new product between two matrices (or tensors) as
\beq
(\bfA \odot \bfB)_{ij}=\sum_{k}\bfA_{ik} \bfB_{jk} \frac{N_k}{N}.
\eeq
It's easy to see that Eq.~\eqref{E:recursive1_sub} and \eqref{E:recursive2_sub} are equivalent to Eq.~\eqref{E:recursive1} and \eqref{E:recursive2} when the product is interpreted as $\odot$. Very much like the ordinary matrix multiplication, $\odot$ is noncummutative, but associative and distributive, which are all that we need for the theory. This shows that Eq.~\eqref{E:rsm_sub} can be proved by identically as the single population case Eq.~\eqref{e:rsm} while interpreting products via $\odot$.

\subsection{Subpopulation theory for higher order correlations}
\label{S:multi_HOC}

The idea in Section~\ref{S:sub_pop_rsm} is exactly how we will develop the subpopulation theory for higher order correlations. Under the interpretation of $\odot$, the relationship among subpopulation motif moments and cumulants
can be written as
\beq
\label{E:k_cumulant_exp_sub}
\M_{n_1,\cdots, n_k}= \sum_{\pi_1,\cdots,\pi_k}
\left(\bigotimes_{i=1}^k
 \prod_{j=2}^{t_i}\K_{B^i_j}\right) 
 \cdot
 \left( \sum_{\chi} \bigodot_{\{i_1,\cdots,i_s\} \in \chi} \K_{B_1^{i_1},\cdots, B_1^{i_s}, \cdot} \right)
\eeq
As before $\pi_i =\{B^i_1,\cdots,B^i_{t_i}\}$ is an ordered partition of $n_i$.  Moreover, $\chi$ is a partition of the set $\{1,\cdots, k\}$ and $\{i_1,\cdots,i_s\}$ is one set of indices that are grouped together under $\chi$. Here $\M_{n_1,\cdots,n_k}$ (for $k\ge 2$) is a $k$ dimensional tensor:  each entry $\M_{n_1,\cdots,n_k}^{\alpha_1,\cdots, \alpha_k}$ represents the frequency of a $k$ brach motif with endpoints in subpopulation $\alpha_1,\cdots,\alpha_k$ respectively. Notably, there is a third type of quantity appearing in Eq.~\eqref{E:k_cumulant_exp_sub}: $\K_{B_1^{i_1},\cdots, B_1^{i_s}, \cdot}$, which is a $s+1$ tensor ($s\ge 2$). The extra dimension (represented by the dot in subscript) comes from specifying the subpopulation of the root node, beside the subpopulation of the endpoints. This is the same situation as for one-branch or chain motifs $\M_{n_1}$ and $\K_{B_j^i}$, which are 2-tensors ($b\times b$ matrices) and should formally be written as $\M_{n_1,\cdot}$ and $\K_{B_j^i,\cdot}$;  we omit the dot for these chains as long as it is clear from the context.  
The big $\odot$ product forms an $k$ tensor out of $|\chi|$ factors, in a way similar to a multivariate trace:
\beq
\label{E:tr_N}
\left(\bigodot_{\{i_1,\cdots,i_s\} \in \chi} \K_{B_1^{i_1},\cdots, B_1^{i_s}, \cdot} \right)^{\alpha_1,\cdots, \alpha_k }=\sum_{\beta=1}^b \frac{N_\beta}{N} \prod_{\{i_1,\cdots,i_s\} \in \chi} 
\K_{B_1^{i_1},\cdots, B_1^{i_s},\cdot}^{\alpha_{i_1},\cdots,\alpha_{i_s},\beta}.
\eeq
As an example, if $\chi$ only contains one partition, that consisting of the set itself, we define $\K_{B_1^{1},\cdots, B_1^{k}}:=\bigodot \K_{B_1^{1},\cdots, B_1^{k}, \cdot}$. It's not hard to see that the meaning of the resulting $s$-tensor $\K_{B_1^{1},\cdots, B_1^{k}}$ is the motif cumulant with specified subpopulations for the endpoints. 

The tensor product ``$\bigotimes_{j=1}^{k}$" in Eq.~\eqref{E:k_cumulant_exp_sub} is simply a weighted version of the ordinary tensor product, that is
\beq
\left(\bigotimes_{j=1}^{k} \bfP_{j} \cdot \bfA \right)^{\alpha_1,\cdots,\alpha_k}
=\sum_{\beta_1,\cdots,\beta_k} 
\left(\prod_{j=1}^k  \frac{N_j}{N} (\bfP_j)_{\alpha_j,\beta_j} \right)
\bfA^{\beta_1,\cdots,\beta_k}.
\eeq

Despite the difference in notation between Eq.~\eqref{E:k_cumulant_exp_sub} and \eqref{E:k_cumulant_exp}, the operations share some basic algebraic properties, namely being associative and distributive --- which are all that is needed in the proof of Eq.~\eqref{E:rsm_hoc}. This allows us to derive, with identical arguments, the subpopulation result:
\begin{corollary}
\label{c:rsm_hoc_sub}
\beq
\label{E:rsm_hoc_sub}
\begin{split}
&\frac{\langle\bfS_{y[k]} \rangle_B}{S_{x[k]}}
=\frac{1}{N^{k-1}} 
\bigotimes_{i=1}^k \left(\bfI-\sum_{n=1}^{\infty} g^n \K_{n} \right)^{-1}
\cdot
\left( \sum_{\chi} 
\bigodot_{\{i_1\} \in \chi} \bfE^{-1} 
\bigodot_{\{i_1,\cdots,i_s\} \in \chi, s>1} 
\left( \sum_{n_1,\cdots,n_s=1}^\infty  g^{\vert n \vert}  \K_{n_1,\cdots, n_s, \cdot} \right)
\right).
\end{split}
\eeq
\end{corollary}

Here, the two ``$\bigodot$" terms are actually two parts of one single product associated with $\chi$, as defined in Eq.~\eqref{E:tr_N}.  Specifically:
\beqrn
&&
\left(
\bigodot_{\{i_1\} \in \chi} \bfE^{-1} 
\bigodot_{\{i_1,\cdots,i_s\} \in \chi, s>1} 
\left( \sum_{n_1,\cdots,n_s=1}^\infty  g^{\vert n \vert}  \K_{n_1,\cdots, n_s, \cdot} \right) \right)
^{\alpha_1,\cdots,\alpha_k}\\
&=&
\sum_{\beta=1}^b \frac{N_\beta}{N} \prod_{\{ i_1\} \in \chi} (\bfE^{-1})_{\alpha_{i_1} \beta} \prod_{\{i_1,\cdots,i_s\} \in \chi} 
\K_{B_1^{i_1},\cdots, B_1^{i_s},\cdot}^{\alpha_{i_1},\cdots,\alpha_{i_s},\beta}
\eeqrn
We emphasize again that all multiplicative operations in the formula above should be interpreted as for $\odot$.  

However, it is also easy to rewrite this expression using only ordinary products, by inserting the diagonal scaling matrix $\bfE$. For example, enumerating the terms for  third order correlation ($k=3$) yields
\beq\label{E:rsm_sub_3rd}
\begin{split}
\langle \bfS_{y[3]} \rangle_B/ S_{x[3]}
&=\frac{1}{N^2}\left(\bfI - \sum_{l=1}^\infty g^l \K_l \bfE\right)^{-1}
 \otimes
\left(\bfI - \sum_{m=1}^\infty g^m \K_m \bfE\right)^{-1}
\otimes
\left(\bfI - \sum_{n=1}^\infty g^n \K_n \bfE\right)^{-1}\\
&\cdot \left(\bfE^{-2}_{[3]}
+ \sum_{l,m=1}^\infty g^{l+m} (\K_{l,m, \cdot}+\K_{l, \cdot, m}+\K_{ \cdot,l,m})
+\sum_{l,m,n=1}^\infty g^{l+m+n} \K_{l,m,n} \right).
\end{split}
\eeq
Here $\bfE^{-2}_{[3]}$ is a diagonal 3-tensor, with $(\bfE^{-2}_{[3]})^{\alpha \beta \gamma}=\delta_{\alpha \beta\gamma}\left( \frac{N_{\alpha}}{N}\right)^{-2}$. $\K_{l,\cdot,m}$ and $\K_{\cdot,l,m}$ are transpositions of the tensor $\K_{l,m,\cdot}$, i.e. $(\K_{l,m,\cdot})^{i_1,i_2,i_3}=(\K_{l,\cdot,m})^{i_1,i_3,i_2}=(\K_{\cdot,l,m})^{i_3,i_1,i_2}$.

\bibliographystyle{aip_title.bst}
\bibliography{hu_trousdale_local_paths}

\begin{thebibliography}{10}
\bibitem{corr_def}
We quantify second order correlations between the activity of two nodes $i,j$
  using the cross-covariance functions, defined by
  $(\mathbf{C}_y(\tau))_{ij}=\mathbf{E}\left[(y_i(t)-\mathbf{E}\left[y_i(t)\right])(y_j(t+\tau)-\mathbf{E}\left[y_j(t+\tau)\right]
  )\right]$ .
\bibitem{ft_sp_def}
Spectral quantities of stochastic processes are technically defined by
  considering first the transform over a finite window - i.e.,
  $\tilde{y}(\omega) = \int_0^{T} dt e^{-2\pi i \omega t}(y_i(t) -
  \mathbf{E}\left[y_i(t)\right])$. {T}he spectrum, for instance, is given by
  ${S}_y(\omega) = \lim_{T\rightarrow\infty}
  \frac{1}{T}\mathbf{E}\left[\bar{\tilde{y}}\tilde{y}(\omega)\right]$ .
\end{thebibliography}

\end{document}